\documentclass[12pt]{iopart}
\usepackage{hyperref}
\usepackage[utf8]{inputenc}
\usepackage{booktabs}
\usepackage[euler]{textgreek}
\usepackage{amsmath, amsthm, amssymb, amsfonts, mathrsfs}
\usepackage[noadjust,compress]{cite}
\usepackage{bm}
\usepackage{bbold}
\usepackage{afterpage}
\usepackage{thm-restate}

\usepackage{float}

\usepackage{tocloft}

\cftsetindents{subsection}{2em}{2em}
\cftsetindents{subsection}{4em}{4em}

\usepackage{pgfplots}
\usepackage{dsfont}
\usepackage{cancel}

\usepackage{enumerate}
\usepackage{csquotes}

\usepackage{listings}
\definecolor{Green1}{rgb}{0.0, 0.50, 0.0}
\definecolor{Graphgreen}{rgb}{0.0, 0.85, 0.0}
\definecolor{GFgreen}{rgb}{0.0, 1.0, 0.0}
\definecolor{Shadowred}{rgb}{0.70, 0.0, 0.0}
\definecolor{Scottred}{rgb}{0.85, 0.0, 0.0}

\newcommand{\pp}{p}
\newcommand{\qq}{q}
\renewcommand{\a}{A}
\renewcommand{\b}{B}
\renewcommand{\P}{\bm{P}}
\newcommand{\Q}{\bm{Q}}

\renewcommand{\S}{\bm{S}}
\newcommand{\T}{\bm{T}}
\newcommand{\A}{\bm{A}}
\newcommand{\B}{\bm{B}}

\newcommand{\Dn}{\Delta_n}
\newcommand{\Ddn}{\Delta_{n,d}}
\newcommand{\Bdn}{\mathcal{B}_{n,d}}

\newtheorem{prop}{Proposition}
\newtheorem{obs}{Observation}

\newtheorem{conj}{Conjecture}
\newtheorem{lemma}{Lemma}

\theoremstyle{remark}

\newtheorem{example}{Example}

\theoremstyle{definition}
\newtheorem{definition}{Definition}

\definecolor{codegreen}{rgb}{0,0.6,0}
\definecolor{codegray}{rgb}{0.5,0.5,0.5}
\definecolor{codepurple}{rgb}{0.58,0,0.82}
\definecolor{backcolour}{rgb}{0.95,0.95,0.92}
 
\lstdefinestyle{mystyle}{
    backgroundcolor=\color{backcolour},   
    commentstyle=\color{codegreen},
    keywordstyle=\color{magenta},
    numberstyle=\tiny\color{codegray},
    stringstyle=\color{codepurple},
    basicstyle=\footnotesize,
    breakatwhitespace=false,         
    breaklines=true,                 
    captionpos=b,                    
    keepspaces=true,                 
    numbers=left,                    
    numbersep=5pt,                  
    showspaces=false,                
    showstringspaces=false,
    showtabs=false,                  
    tabsize=2
}
 
\usepackage[small,bf]{caption}
\newcommand{\ket}[1]{\left|#1\right\rangle}
\newcommand{\bra}[1]{\left\langle#1\right|}

\DeclareMathOperator{\prob}{p}

\newcommand{\dyad}[1]{| #1\rangle \langle #1|}
\newcommand{\ot}[0]{\otimes}
\newcommand{\one}[0]{\mathds{1}}
\renewcommand{\a}{S}

\newcommand{\R}{\mathds{R}}
\newcommand{\C}{\mathds{C}}

\newcommand{\conv}{\rm conv}
\renewcommand{\t}{t}
\renewcommand{\*}{*}

\newtheorem{theorem}{Theorem}
\usepackage{array}
\usepackage{rotating}
\usepackage{tabularx}
\usepackage{geometry}
\usepackage{pifont}

\usepackage{setspace}

\usepackage{soul,xcolor}

\begin{document}

\setstcolor{red}

\title[]{Discrete dynamics in the set of quantum measurements}

\author{Albert Rico${}^1$ \& Karol \.Zyczkowski${}^{1,2}$}

\address{${}^1$Faculty of Physics, Astronomy and Applied Computer Science, Institute of Theoretical Physics, Jagiellonian University,
30-348 Krak\'{o}w, 
Poland
}
\address{${}^2$Center for Theoretical Physics, Polish Academy of Sciences, 02-668 Warszawa, Poland 
}

\quad\quad\quad\,\,\,\,\, 
\today

\vspace{10pt}

\hspace{50pt}\href{mailto:albertrico23@gmail.com;albert.andres@uj.edu.pl}{albert.andres@uj.edu.pl}

\begin{abstract}

A quantum measurement, often referred to as positive operator-valued measurement (POVM), 
is a set of positive operators $P_j=P_j^\dag\geq 0$ summing to identity, $\sum_jP_j=\one$. 
This can be seen as a generalization of a probability distribution of positive real numbers summing to unity, whose evolution is given by a stochastic matrix. 
We describe transformations in the set of quantum measurements induced by {\em blockwise stochastic matrices}, composed of
positive blocks
that sum columnwise to identity,
using the notion of {\em sequential product} of matrices. We show that such transformations correspond to a
sequence of quantum measurements.  
Imposing additionally the dual condition that
the sum of blocks in each row is equal to identity 
we arrive at blockwise bistochastic matrices (also called {\em quantum magic squares}). Analyzing their dynamical properties, we formulate our main result: a quantum analog of the Ostrowski description of the classical Birkhoff polytope, which introduces the notion of majorization between quantum measurements.
Our framework provides a dynamical characterization of the set of blockwise bistochastic matrices and establishes a resource theory in the set of quantum measurements.

\end{abstract}

\newpage
\section{Introduction}\label{sec:Intro}
 
Quantum theory can be regarded as a generalization of classical probabilistic theory, where classical phenomena emerge as a special case of quantum phenomena under certain constraints. Understanding the similarities and differences between quantum and classical theory is of both fundamental and applied interest. Specifically, quantum measurements are based on mathematical postulates that yield outcome probability distributions
~\cite{watrous_2018}, revealing certain properties of a quantum system
\cite{nielsenChuang2010quantum}. 

Although extensive research has been devoted to exploring the relationship between quantum measurements and probabilistic experiments~\cite{gudder2014QPT},
the comparison of the dynamical properties of discrete quantum measurements and discrete probability distributions 
requires further studies.
This work aims to identify the conditions under which standard properties of the evolution of discrete probability distributions apply analogously to the evolution of discrete quantum measurements.

In classical probability theory, a discrete probability distribution is a column vector, $\pp=(p_1,...,p_n)^T$, with nonnegative entries $p_i$ which sum to unity. The set of probability vectors $\Delta_n$ is known as the {\em probability simplex}. Discrete evolution of a probability vector $\pp$ is described by multiplication by a stochastic matrix $S$ with nonnegative entries $s_{ij}$ and columns $s_j$ summing to unity. Physically, this can be understood as follows: first, an experimental measurement with outcome distribution $\pp$ (e.g. rolling dices) is performed; depending on the outcome $j$, a second measurement is performed with outcome probability distribution $s_j$, where each outcome $i$ occurs with conditional probability $s_{ij}$. The marginal probability distribution after both measurements is $\qq=(q_1,...,q_n)^T$ with
\begin{equation}\label{eq:ClassicConditional}
    q_i = \sum_{j=1}^n s_{ij}p_j\,.
\end{equation}
Discrete probabilistic systems whose evolution is given by a stochastic matrix via Eq.~\eqref{eq:ClassicConditional} are known as {\em Markov chains}, and have the property that the state of the system at each time step depends only on the state of the system at the previous time step~\cite{norris1998markov}.

A particular dynamics in the space of classical probability distributions is induced by bistochastic matrices, with nonnegative entries which sum to unity for both rows and columns. The set of bistochastic matrices of a certain size $n$ has been characterized in two equivalent ways:
\begin{enumerate}
    \item {\em Birkhoff-von Neumann characterization}~\cite{birkhoff1946tres}: The set of bistochastaic matrices of size $n$, also known as the Birkhoff polytope, is the convex hull of all permutation matrices of size $n$. 
    \item {\em Ostrowski characterization}~\cite{ostrowski1952quelques}: A square matrix of size $n$ with nonnegative entries is bistochastic, if and only if, one has the majorization relation $p\succ Bp$ for any probability vector $p\in\Delta_n$.
\end{enumerate}
The majorization relation $p\succ q$ between two vectors $p,q\in\Delta_n$ means that
\begin{equation}
    \sum_{j=1}^k p_j\geq \sum_{j=1}^k q_j
\end{equation}
for all $k=1,2,...,n$ where the entries $p_j$ are sorted in non-increasing order~\cite{marshall1979inequalities,Bhatia_MatAn1997}. Due to the characterization of Ostrowski~\cite{ostrowski1952quelques,marshall1979inequalities} summarized in item (ii), bistochastic matrices induce majorization relations between probability distributions and a corresponding resource theory~\cite{ChitambarGourQRTheos_2019}. 

Similar ideas to the ones describing the probability simplex can be used in quantum theory, in various ways. The most common approach is to consider density matrices (positive semidefinite normalized operators $\varrho=\varrho^\dag\geq 0$) as coherent generalizations of probability distributions. Discrete dynamics in the set $\Omega_d$ of density matrices of order $d$ is described by quantum operations $\mathcal{E}$ (completely positive and trace preserving) represented by Kraus operators $K_i$ in the form of 
\begin{equation}
    \mathcal{E}(\varrho)=\sum_i K_i\varrho K_i^\dag\,,
\end{equation}
where the {\em effects} $K_i^\dag K_i$ sum to identity, $\sum_i K_i^\dag K_i=\one$. In this work we use an alternative approach in the spirit of Gudder~\cite{Gudder2001sequential,Gudder08,gudder2023properties}, in which probability vectors are generalized to vectors whose components are effects, $P_i:=K_i^\dag K_i$, summing to identity~\cite{GueriniBProbVAndBBist_2018,Gemma2020}.

The goal of this paper is to understand whether certain dynamical properties of the probability simplex $\Delta_n$ hold analogously for the set of quantum measurements, when the agents of the evolution are elements of their own set. We start from a vector description of quantum measurements~\cite{GueriniBProbVAndBBist_2018} and define a product between these objects, which is based on matrix convex combination of effects~\cite{Gudder2001sequential,paulsen_MatConvex_2003}. This product can be understood as a concatenation of quantum measurements, where transformations of measurement are induced by objects {\em within} their own set. Thus transformations correspond to sequential L\"{u}ders measurements, where more general quantum instruments~\cite{davies1970operational} are absorbed into sequential measurement effects. This can be seen as a case of channel transformation ~\cite{Chiribella_Supermaps2008,Quartic-Karol_2008,Gour_Supermaps2019,RegulaRuyjiFLimDistChan_2021} applied to quantum measurements. In particular, bistochastic matrices have been generalized over an operator algebra in several contexts
~\cite{Benoist_BlockWiseObjs2017,GueriniBProbVAndBBist_2018,Gemma2020,GemmaMagic2022,bluhm2023polytope,epperly2022SpectrahedraQMS,hoefer2022GamesQMS,evert2023freeQMS,brannan2023NSbicorrQMS},
as blockwise square matrices where positive semidefinite entries sum to identity for rows and columns. These objects are called {\em blockwise bistochastic matrices}~\cite{Benoist_BlockWiseObjs2017} or {\em quantum magic squares}~\cite{Gemma2020}, and constitute a very special case of blockwise stochastic matrices inducing dynamics in our framework.

This paper is organized as follows. In section~\ref{sec:BPVandPOVMS} we introduce the set of quantum measurements as blockwise probability vectors, and discuss its geometry and generalized convex hulls. In section~\ref{sec:DiscDynQPS} we introduce blockwise stochastic matrices and define a product between quantum measurements based on a generalized convex combination of effects. We discuss its physical meaning in terms of conditional measurements and measurement compatibility. In section~\ref{sec:BlockBistoch} we introduce a generalized notion of majorization between quantum measurements and use it to establish a dynamical characterization of blockwise bistochastic matrices analogous to the Ostrowski characterization of the Birkhoff polytope~\cite{birkhoff1946tres,ostrowski1952quelques}. In section~\ref{sec:ResThryPOVMs} we show that a resource theory of quantum measurements arises from our results, as a quantum analog of the majorization resource theory in the probability simplex.

\section{The set of quantum measurements}\label{sec:BPVandPOVMS}
A discrete classical probabilistic measurement is described by a column probability vector $\pp=(p_1,...,p_n)^T$ with probabilities $p_j\geq 0$ normalized to $\sum_{j=1}^n\pp_j=1$. 
In the quantum setup, a measurement is a set of positive operators $\{P_j\}$ summing to the identity. As a natural generalization of the classical case, quantum measurements on a system of dimension $d$ can be described as follows~\cite{GueriniBProbVAndBBist_2018}.
\begin{definition}[Blockwise probability vector]\label{def:BlockProbVect}
A \emph{blockwise probability vector} is a column vector
\begin{equation}
\bm{P}=
    \begin{pmatrix}
    P_1 \\
    \vdots \\
    P_n
    \end{pmatrix}\,,
\end{equation}
with $n$ components $P_j$ being Hermitian, positive semidefinite matrices of order $d$, $P_j\geq 0$, satisfying the identity resolution
\begin{equation}\label{eq:StrongNormalization}
   \sum_{j=1}^n P_j=\one_d. 
\end{equation}
\end{definition}
\noindent Thus a blockwise probability vector $\P=(P_1,...,P_n)^\dag$ corresponds to a generalized quantum measurement, whose effects $P_j$ play the role of vector components over an operator algebra. The classical notion of probability vectors with $n$ real nonnegative components is recovered for $d=1$. This can be obtained by taking the expectation value of each effect with respect to a quantum state $\varrho$ acting on $\C^d$, which gives the probability $p_j=\prob(j)=\tr(P_j\varrho)$ of obtaining the outcome $j$. A similar notion of blockwise vector was introduced in~\cite{Gudder08}.

The probability simplex $\Delta_n$ of $n$-point probability vectors 
\begin{equation}
    \Delta_n=\big \{\sum_{j=1}^n {v_j}p_j \big \}_{p_j\geq 0,\,\sum_jp_j=1} \subset\R^{n-1}\,,
\end{equation}
is given by the convex hull of the vectors $v_j$,
\begin{equation}\label{eq:CartesianVectors}
    {v_1}=(1,0,...,0)^T,\,{v_2}=(0,1,0,...,0)^T,\,...\,,\,{v_n}=(0,0,...,1)^T
\end{equation}
which form its extremal points. 
The extremal points of the set of quantum measurements with $n$ effects in dimension $d$ induce a more involved geometry~\cite{arveson1969subalgebras,
VirmaniExtrLocPOVMs_2003,
ConvexPOVMs_Dariano2005,
pellonpaa2011complete,
holevo2011probabilistic,
SentisDecompPOVMs_2013,
OszmaniecExtrPOVMs_2016,
MartinezPOVMOptByExtr_2020}. Here we present two generalizations of the standard convex geometry of the probability simplex.

Following~\cite{GueriniBProbVAndBBist_2018,GemmaMagic2022}, we consider a decomposition of a blockwise probability vector $\P$ in the form of
\begin{equation}\label{eq:TensProdConvComb}
    \P = \sum_{j=1}^n {v_j}\ot P_j\,,
\end{equation}
where $P_j\geq 0$ and $\sum_{j=1}^n P_j=\one_d$. In this sense, the set of blockwise probability vectors can be seen as a generalization of the probability simplex $\Dn$, as it can be obtained from the vectors $v_j$ as
\begin{equation}
     \big \{\sum_{j=1}^n {v_j}\otimes P_j \big \}_{P_j\geq 0,\,\sum_jP_j=\one} \subset\R^{n-1}\ot(\C^{d}\times\C^{d})\,.
\end{equation}
The set $\Ddn$ of blockwise probability vectors with $n$ blocks of size $d$ each contains the central point
\begin{equation}\label{eq:UniformBPV}
    \bm{P_u}=(\one_d/n,...,\one_d/n)^\dag\,
\end{equation}
and $d$ particular extremal points
\begin{equation}\label{eq:ExtremalBPV}
    \bm{V_j} = {v_j}\otimes\one_d\,,
\end{equation}
known as {\em fuzzy} POVMs~\cite{buscemiCompSharpResTQM_2024}, which generate $\Ddn$ through the non-standard convex combination of Eq.~\eqref{eq:TensProdConvComb}. 

Another important generalization of convexity of the standard probability simplex $\Delta_n$
can be illustrated as follows: if two probability vectors $p$ and $q$ are contained
in $\Delta_n$, then their convex combination
$c=\sqrt{a} p 
       \sqrt{a} + \sqrt{1-a} q\sqrt{1-a}$
also belongs to $\Delta_n$ for any $a\in[0,1]$.
Here a positive number $a$ is written as a product of its square roots to pave the road for 
the following generalization.
The set $\Ddn$ of blockwise probability vectors is {\em matrix convex} 
\cite{effros1997matrix,kriel2019introduction,klep2022facial}: given two such vectors $\P$ and $\Q$ 
their matrix-convex combination
${\bm{C}}$, with components
\begin{equation}
\label{eq:matrix_conv}
C_j=\sqrt{A} UP_j U^{\dagger} 
       \sqrt{A} + \sqrt{\one-A} V Q_j V^{\dagger}\sqrt{\one-A},
\end{equation}
belongs to $\Ddn$ for any
any unitary $U,V\in U(d)$
and any Hermitian matrix $A$ such that $0\le A \le \one$.
To show this it is enough to realize
that all components are positive, $C_j\ge 0$,
and normalized, $\sum_{j=1}^n C_j =\one$.       

Understanding the convex geometric structure of the set of quantum measurements of arbitrary dimension $d$ and number of effects $n$ is an ongoing challenge~\cite{ConvexPOVMs_Dariano2005,SentisDecompPOVMs_2013,JenExtrmGenPOVMs_2013}, this task is much simpler in the case of $n=2$, for the set of two-effect measurements $\Delta_{2,d}$. Its vectors consist of two effects, $\bm{P}=(P,\one-P)^\dag$, and thus it is fully described by the set of positive semidefinite matrices of order $d$ having eigenvalues between 0 and 1,
\begin{equation}
    \Delta_{2,d}\approx\{P:\quad 0\leq P\leq\one_d\}.
\end{equation}
This means that $\Delta_{2,d}$ is defined by a fraction of the positive cone, which is obtained by the convex hull of pairs of orthogonal projectors.

Let us focus on the simplest nontrivial set of blockwise probability vectors $\Delta_{2,2}$, which describes a single-qubit measurement with two effects $P$ and $\one-P$, and denote as $\lambda_+$ and $\lambda_-$ the maximal and minimal eigenvalues of $P$ respectively.
Since $P$ can be viewed as a quantum state up to a normalization factor $t:=\tr(P)=\lambda_++\lambda_-$, the difference of eigenvalues $\lambda_+-\lambda_-$ can be viewed as the length of a generalized Bloch vector $\Vec{\tau}$ with components $\tau_i=\tr(P\sigma_i)$ over the Pauli basis $\{\sigma_i\}$,
\begin{equation}
    \tau:=|\Vec{\tau}|=\lambda_+-\lambda_-=2\lambda_+-\t\,.
\end{equation}
Here the length $\tau$ is bounded between 0 and 1, since one has $0\leq\lambda_+\leq 1$, and the variable $t=\tr(P)$ is bounded between $0\leq \t\leq 2$.

Note that the set $\Delta_{2,2}$ has the structure of two anti-parallel 4-dimensional cones glued at $\t=\tr(P)=1$ along the Bloch ball $\Omega_2\subset\R^3$ (see Figure~\ref{fig:CheeseM=2}), which is formed by the convex hull of the Bloch sphere $\partial\Omega_2\subset\R^2$ containing all single-qubit pure states. The extremal points of the set $\Delta_{2,2}$ are the $\hat{0}$ and $\one$ operators, located in the apexes of the hypercone with $\t=0$ and $\t=2$ respectively, together with the sphere $\partial\Omega_2$. Therefore, $\Delta_{2,2}$ can be written as a convex hull of these sets,
\begin{equation}\label{eq:Delta22ConvHull}
    \Delta_{2,2}=\conv\Big (\{\hat{0}\}\cup\partial\Omega_2\cup\{\one\}\Big )=\conv\Big ( \{\hat{0},U\dyad{\psi}U^\dag,\one\}_U\Big )\,,
\end{equation}
where $\dyad{\psi}$ is a fixed rank-one projector and the set is obtained by all possible rotations by unitary matrices $U\in\mathcal{U}(2)$. Therefore, the extremal points of $\Delta_{2,2}$ are composed of all projectors of rank $0$, $1$ and $2$.

\begin{figure}[h]
    \centering
    \includegraphics[scale=1]{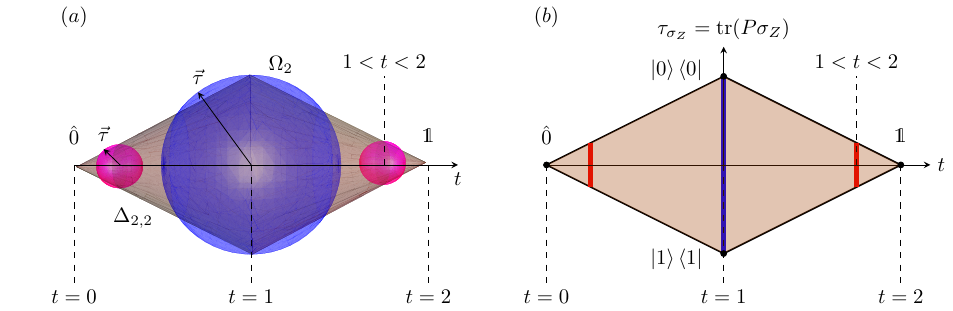}
    \caption{{\bf Visualization of the set of single-qubit measurements with two effects, $\Delta_{2,2}$}. ({\it a}) The dark set $\Delta_{2,2}$ corresponds to the intersection of two antiparallel 4-dimensional hypercones with axial coordinate $\t=\tr(P)$ and 3-dimensional spherical sections parameterized by $\vec{\tau}$ (blue and red). The blue section defined by $\t=1$ defines the Bloch ball $\Omega_2$. Blockwise vectors $\P$ with $\tr(P)\neq 1$ are described by the two distinct red sections located at $1<\t<2$ and its complementary.
    ({\it b}) The dark set corresponds to the two-dimensional projection of $\Delta_{2,2}$ onto the plane with coordinates $\t$ and $\tau_{\sigma_Z}=\tr(P\sigma_Z)$, so that its blue and red sections are lines of length $2|\t-1|$ parameterized by $\tau_{\sigma_Z}$. The four extremal points (black dots) of this projection of $\Delta_{2,2}$ are $P=\hat{0}$ and $P=\one$, in the apexes, and $\ket{0}\bra{0}$ and $\ket{1}\bra{1}$, in the edges of the intersection of the two hypercones.}\label{fig:CheeseM=2}
\end{figure}
    
As a result of Eq.~\eqref{eq:Delta22ConvHull}, the hypervolume of the set $\Delta_{2,2}$ represented in Figure~\ref{fig:CheeseM=2} is the hypervolume of two 4-dimensional hypercones of height $h:=\underset{P}{\max}\,\,\,t/2=1$ and spherical base of radius $a:=\underset{P}{\max}\,\,\,\tau=1$, namely
\begin{equation}
    V(\Delta_{2,2})=2\times\frac{1}{4}h\Big (\frac{4}{3}\pi a^3\Big )=\frac{2\pi}{3}.
\end{equation}
This implies that the ratio between hypervolumes of $\Delta_{2,2}$ and the set of positive semidefinite operators $P$ satisfying $0\leq\tr(P)\leq 2$ is $1/8$. In other words, assigning to $P$ a flat Hilbert-Schmidt measure in the space of positive semidefinite operators bounded by $\tr(P)\leq 2$, $\P=(P,\one-P)$ will be a valid measurement in $\Delta_{2,2}$ with probability $1/8$.

\section{Discrete dynamics induced by blockwise stochastic matrices}\label{sec:DiscDynQPS}
Consider a stochastic matrix $S=(s_{ij})$ of size $n$, so that $s_{ij}\geq 0$ and $\sum_{i}s_{ij}=1$. Such a matrix can be considered as a collection of $n$ independent probability vectors $s_j=(s_{1j},...,s_{nj})^T$,
\begin{equation}
    S = 
    \begin{pmatrix}
        s_1 & \dots & s_n
    \end{pmatrix}
    \in{\Dn}^{\times n}\,,
\end{equation} 
where $s_{ij}=\prob(i|j)$ is the probability that an outcome $i$ is obtained given that a classical measurement $j$ was performed. As sketched in Section~\ref{sec:Intro}, $S$ determines the discrete evolution from $p\in\Dn$ to $q\in\Dn$ via Eq.~\eqref{eq:ClassicConditional}, where each convex combination $q_i$ is the outcome probability after concatenating two conditional measurements. 
Here we will consider the quantum analog of these notions.

\subsection{Blockwise stochastic matrices and their blockwise product}
A collection of $n$ independent measurements with $n$ outcomes can be described by a Cartesian product of $n$ blockwise probability vectors $\S_j$ with $n$ effects $\{S_{1j},...,S_{nj}\}$. This suggests considering the following quantum version of a stochastic matrix.
\begin{definition}[Blockwise stochastic]\label{def:BlockStoch}
Let $\S$ be a square matrix of size $dn\times dn$ composed of $n^2$ blocks $S_{ij}$ of size $d\times d$ each. We call $\bm{S}$ \emph{blockwise stochastic} if 
\begin{enumerate}
    \item Its entries $S_{ij}$ are hermitian,  
    positive semi-definite matrices, $S_{ij}\geq 0$, and
    \item Its $n$ block-columns of size $d$ sum to the identity, $\sum_iS_{ij}=\one$.
\end{enumerate}
\end{definition}
\noindent Since a blockwise stochastic matrix $\S$ is a Cartesian product of $n$ blockwise probability vectors, $\S=(\S_1,...,\S_n)$, 
their set can be written as ${\Delta_{n,d}}^{\times n}$. Similarly as for blockwise probability vectors, the standard notion of a stochastic matrix is recovered by setting $d=1$, when taking the expectation value $\tr(S_{ij}\varrho) = \prob(i|j)$, which gives the probability of obtaining an outcome $i$ provided a measurement $\S_j$ was performed.

To establish interconversion rules in the set of quantum measurements, we will consider the quantum version of Eq.~\eqref{eq:ClassicConditional}, sketched in Figure~\ref{fig:2MeasurementsGraphicRep}. First a quantum measurement described by a vector $\P\in\Ddn$ is performed. If the outcome is $j$ (i.e. the effect $P_j$ has been applied), then a second measurement $\S_j$ is performed. While in the most general case $P_j$ is an instrument acting on $S_{ij}$, here we will restrict ourselves to the L\"{u}ders assumption of a non-filtering setup~\cite{BuschTheQTofMeas_1991}. Here the state after measuring $\P$ on an initial state $\varrho$ reads $\varrho'=\sqrt{P_j}\varrho\sqrt{P_j}/\tr(P_j\varrho)$, where positive semidefinite square root $\sqrt{P_j}=P_j^{1/2}$ is the Kraus operator associated to the effect $P_j=\sqrt{P_j}\sqrt{P_j}$ corresponding to an outcome $j$; and similarly, the state after the second measurement reads $\varrho''=\sqrt{S_{ij}}\varrho'\sqrt{S_{ij}}/\tr(\varrho'S_{ij})$ in the event of obtaining an outcome $i$. Then the probability distribution after this sequence of two measurements can be obtained by an effective quantum measurement, whose effects $\{Q_1,...,Q_n\}$ are given by matrix convex combination~\cite{paulsen_MatConvex_2003} of $S_{i1},...,S_{in}$ over $P_1,...,P_n$, via $Q_i = \sum_{j=1}^n\sqrt{P_j}S_{ij}\sqrt{P_j}\,$ -- see Eq.~\eqref{eq:IntroPosProd} below. That is, the probability of obtaining the outcome $i$ reads \begin{equation}\label{eq:probQi}
\prob(i)=\tr(Q_i\varrho)=\sum_{j=1}^n\tr(\sqrt{S_{ij}}\sqrt{P_j}\varrho\sqrt{P_j}\sqrt{S_{ij}})=\sum_{j=1}^n\prob(j,i)\,,
\end{equation}
where $\prob(j,i)$ is the probability of obtaining $j$ in the first measurement and $i$ in the second one in the setup described above. Note that this includes the more general case of an instrument~\cite{davies1970operational} $\Phi_j(\sqrt{P_j}\rho\sqrt{P_j})$ sequenced by a measurement with set of effects $\{E_{ij}\}_i$ for a fixed outcome $j$, as the channel $\Phi_j$ with adjoint $\Phi_j^*$ can be absorbed into $S_{ij}:=\Phi_j^*(E_{ij})$.

In the outcome effects $Q_i$, each term $\sqrt{P_j}S_{ij}\sqrt{P_j}$ can be seen as the {\em sequential product}~\cite{Gudder2001sequential,gudder2023properties} or {\em star product}~\cite{leiferStarProd_2007,leiferSpeckensStarProd_2013} between the effects $P_j$ and $S_{ij}$. One verifies that indeed $Q_i=Q_i^\dag\geq 0$ and $\sum_{i=1}^nQ_i=\one_d$, so that $\Q$ defines a quantum measurement.
\begin{figure}[tbp]
    \centering
    \includegraphics[scale=1]{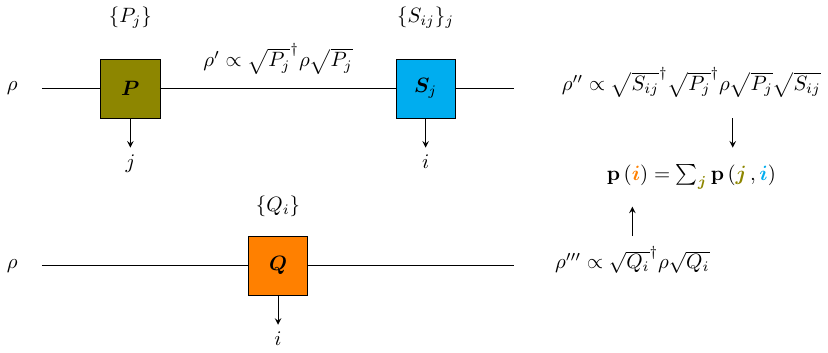}
    \caption{{\bf Circuit description of sequential measurements and measurement interconversion.} 
    {\em Top}: The system is initialized in a quantum state $\varrho$. A first measurement $\P$ with effects $P_j$ is performed (in green), and depending on the outcome $j$, a second measurement $\S_{j}$ with effects $S_{ij}$ is performed (in blue). {\em Bottom}: A single measurement $\Q=\S\*\P$ with effects $Q_i=\sum_j\sqrt{P_j}S_{ij}\sqrt{P_j}$ is performed (in orange) on an initial quantum state $\varrho$. Both scenarios give equivalent marginal probability distributions of an outcome $i$.}

    \label{fig:2MeasurementsGraphicRep}
\end{figure}
Therefore, the discrete evolution of quantum measurements introduced in this work carries the following generalization of the matrix product.
\begin{definition}[Blockwise product]\label{def:PositiveProduct}
Let $\bm{A}$ and $\bm{B}$ be two matrices $\bm{A}=(A_{ij})\in\C_{dn}\times \C_{dn'}$ and $\bm{B}=(B_{ij})\in\C_{dn'}\times \C_{dn''}$ composed of $n\times n'$ and $n'\times n''$ positive semidefinite blocks of size $d$, $A_{ij},B_{ij}\in\C_{d}\times \C_{d}$. We define the {\it blockwise product},
\begin{equation}\label{eq:DefPosProd}
\begin{aligned}
(\A\*\B)_{ik}&=\Big (\sum_{j}\sqrt{B_{jk}}A_{ij}\sqrt{B_{jk}}\Big )\in\C_{dn}\times \C_{dn''}\,.
\end{aligned}
\end{equation}
\end{definition}
One verifies that the blockwise product preserves positivity at each block $(\A\*\B)_{ij}$. Further properties of this product are presented in~\ref{App:PropsProduct}. In natural analogy to the evolution of classical measurements, the transformed effects $Q_i$ discussed above are obtained with the blockwise product of $\bm{S}$ and $\P$,
\begin{equation}\label{eq:IntroPosProd}
    \bm{S}\*\P=
    \begin{pmatrix}
        S_{11} & \dots & S_{1n} \\
        \vdots & \ddots & \\
        S_{n1} & & S_{nn}
    \end{pmatrix}
    \*
    \begin{pmatrix}
        P_1 \\
        \vdots \\
        P_n
    \end{pmatrix}
    =
    \begin{pmatrix}
        \sum_{j=1}^n\sqrt{P_j}S_{1j}\sqrt{P_j} \\
        \vdots \\
        \sum_{j=1}^n\sqrt{P_j}S_{nj}\sqrt{P_j}
    \end{pmatrix}
    =
    \begin{pmatrix}
        Q_1 \\
        \vdots \\
        Q_n
    \end{pmatrix}
    =\Q\,,
\end{equation}
where it is straightforward to note that $\sum_{i=1}^nQ_i=\one$ as $\S$ is blockwise stochastic.
\begin{example}
Consider the Pauli projective measurements characterized by blockwise probability vectors
\begin{equation}
\P_Z=
\begin{pmatrix}
    P_{Z+} \\
    P_{Z-}
\end{pmatrix}\,,\quad
\P_X=
\begin{pmatrix}
    P_{X+} \\
    P_{X-}
\end{pmatrix}\,\quad
\P_Y=
\begin{pmatrix}
    P_{Y+} \\
    P_{Y-}
\end{pmatrix}\quad
\text{and}\quad
\P_\one=
\begin{pmatrix}
    \one/2 \\
    \one/2
\end{pmatrix}\,,
\end{equation}
where each $P_{i\pm}$ is the rank-1 projector with eigenvalue $\pm 1$ on the Pauli matrix $\sigma_i$, e.g.
\begin{equation}\label{eq:PXdef}
P_{X\pm}=\frac{1}{2}
\begin{pmatrix}
    1 & \pm 1 \\
    \pm 1 & 1
\end{pmatrix}.
\end{equation}
Three possible measurements can be obtained up to changes of the basis, depending on the alignment between the first and second measurement:
\begin{align}
    (\P_Z,\P_Z)\*\P_Z &=\P_Z\quad,\\
    (\P_Z,\P_Z)\*\P_X &=(\P_Z,\P_X)\*\P_Y=\P_\one\quad\text{or}\\
    (\P_Z,\P_X)\*\P_X &=
    \begin{pmatrix}
        P_{X+}/2 \\
        \one - P_{X+}/2
    \end{pmatrix}=
    \frac{1}{2}
    \begin{pmatrix}
        P_{X+} \\
        P_{X-}
    \end{pmatrix}+
    \frac{1}{2}
    \begin{pmatrix}
        0 \\
        \one
    \end{pmatrix}
    \,.
\end{align}
\end{example}

\subsection{Dynamics in the space of quantum measurements}
Similarly as for stochastic matrices, the set of blockwise stochastic matrices is closed under the blockwise product $*$ and has a unit element $\one$. This can be formulated as follows.
\begin{prop}\label{prop:PreserveStochastic}
The set of blockwise stochastic matrices equipped with the blockwise product, $\{{\Ddn}^{\times n},*\}$, forms an unital magma, i.e. a closed set with binary operation and neutral element.
\end{prop}
\begin{proof}
The neutral element is the identity $\one_{dn}$, since $\one * \bm{S}=\bm{S} * \one=\bm{S}$ for any blockwise stochastic matrix $\bm{S}$. Moreover, given two blockwise stochastic matrices $\bm{S},\bm{T}\in{\Ddn}^{\times n}$, their blockwise product $\bm{R}=\bm{S}\*\bm{T}$ is also blockwise stochastic.
This is because the product $*$ preserves the positivity of blocks, and the sum over all elements of the $k$-th column of $\bm{R}$ reads
\begin{equation}
\sum_iR_{ik}=
\sum_{j}\sqrt{T_{jk}}\Big (\sum_iS_{ij}\Big )\sqrt{T_{jk}}=
\one\,.
\end{equation}
\end{proof}
In the classical case, one verifies that given two probability vectors $\pp,\qq\in\Dn$, there exists a stochastic matrix $\a_{\pp\rightarrow\qq}$ such that $\qq=\a_{\pp\rightarrow\qq}\cdot\pp$~\cite{Bhatia_MatAn1997}. One can for example choose $\a_{\pp\rightarrow\qq}=(\qq,...,\qq)$ such that all columns are equal to $\qq$. Although for a full analysis one needs to consider the blockwise product $*$ and its dual version $*^\dag$ -- see~\ref{App:PropsProduct}, below we show that this property does not hold for POVM interconversion via Definition~\ref{def:PositiveProduct}.

\begin{prop}\label{prop:Interconv_Q_Simplex}
There are quantum measurements $\P\in\Ddn$ and $\Q\in\Ddn$ for which there exists no blockwise stochastic matrix $\bm{S}\in{\Delta_{n,d}}^{\times n}$ transforming $\P$ to $\Q$ via $\bm{Q}=\S\*\P$.
\end{prop}
\begin{proof}
Consider the blockwise probability vectors $\P=(P_{Z+},P_{Z-})^\dag$ and $\Q=(P_{X+},P_{X-})^\dag$, where $P_{X\pm}$ are given by Eq.~\eqref{eq:PXdef} and $P_{Z\pm}$ are diagonal projectors. For any blockwise stochastic matrix
\begin{equation}
\S=
\begin{pmatrix}\label{eq:StochN=2}
    S & T \\
    \one-S & \one-T
\end{pmatrix}\,,
\end{equation}
one has
\begin{equation}
\S\*\P=
\begin{pmatrix}
    P_{Z+}SP_{Z+} + P_{Z-}TP_{Z-} \\
    P_{Z+}(\one-S)P_{Z+} + P_{Z-}(\one-T)P_{Z-}
\end{pmatrix}\,,
\end{equation}
where all blocks are diagonal. Therefore one cannot obtain a blockwise probability vector with effects that are not diagonal, e.g. there exists no blockwise stochastic matrix $\S$ such that $\S\*\P=(P_{X+},P_{X-})^\dag$.
\end{proof}

\subsection{Regions accessible by the blockwise product}
Proposition~\ref{prop:Interconv_Q_Simplex} shows that not all transformations between quantum measurements can be done by conditional concatenation of measurements which simulate the output probabilities. This raises the following question: {\it What is the set of quantum measurements which can be obtained from an arbitrary blockwise probability vector $\P$ by blockwise product with a blockwise stochastic matrix $\S$?} To answer this question, first note that the set of quantum measurements is convex and characterized by its extremal points, which contain projective measurements 
$\P=(P_1,...,P_n)^\dag\in\Ddn$ with $P_j^2=P_j$. Based on this fact, the following proposition bounds the set of blockwise probability vectors which can be obtained by a fixed measurement described by $\P\in\Ddn$.
\begin{prop}
    Let $\P\in\Ddn$ be a blockwise probability vector. The following interconversion rules hold.
    \begin{enumerate}
        \item The set of measurements $\Q$ that can be obtained from $\P$ with the blockwise product $\*$ through a blockwise bistochastic matrix is convex.
        \item The extremal points of this set are obtained by blockwise stochastic matrices, the columns of which form extremal measurements.
    \end{enumerate} 
\end{prop}
\begin{proof}
The proof relies on the fact that the blockwise product $\*$ is partially distributive,
\begin{equation}\label{eq:Distrib\*}
    (\S+\S')\*\P=\S\*\P+\S'\*\P\,.
\end{equation}

For (i) we need to show that if $\Q$ and $\Q'$ can be attained from $\P$ via $\S$ and $\S'$, then $\alpha\Q+(1-\alpha)\Q'$ with $0\leq\alpha\leq 1$ can also be attained from $\P$ with a blockwise stochastic matrix. One has by assumption $\Q=\S\*\P$ and $\Q'=\S'\*\P$, and therefore by Eq.~\eqref{eq:Distrib\*} we have
\begin{equation}
    \alpha\Q+(1-\alpha)\Q' = \big (\alpha\S+(1-\alpha)\S'\big )\*\P
\end{equation}
where $\T:=\alpha\S+(1-\alpha)\S'\in{\Ddn}^{\times n}$ is blockwise stochastic, since the set of blockwise stochastic matrices is convex.

For (ii) we need to show that the extremal points of the set of vectors attainable from a blockwise probability vector $\P$ via $\Q=\S\*\P$ are obtainable from a matrix $\S\in{\Ddn}^{\times n}$ such that all its columns define extremal points in $\Ddn$. To see this, it is enough to show that if $\Q$ does not describe an extremal measurement, namely $\Q=\alpha\Q'+(1-\alpha)\Q''$ with $\Q'\neq\Q''\in\Ddn$ and $0\leq\alpha\leq 1$, then $\S=\alpha\S'+(1-\alpha)\S''$. Indeed, this holds due to Eq.~\eqref{eq:Distrib\*}. 
\end{proof}
\begin{example}
\begin{figure}[tbp]
    \centering
    \includegraphics[scale=1]{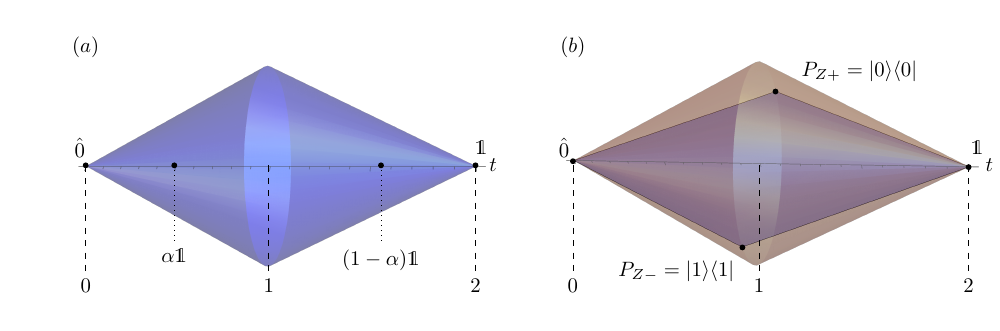}
    \caption{\textbf{Regions of the set of measurements $\Delta_{2,2}$  accessible with the blockwise product.} As in Fig.~\ref{fig:CheeseM=2}, the set of vectors $\P=(P,\one-P)$ defines two anti-parallel cones along the coordinate $\t=\tr(P)$. ($a$) All quantum measurements are accessible from $\P_\alpha=(\alpha\one,(1-\alpha)\one)^\dag$ with a blockwise stochastic matrix.
    ($b$) The subset of measurements accessible from $\P_Z=(P_{Z+},P_{Z-})^\dag=(\dyad{0},\dyad{1})^\dag$ is a fragment of a plane obtained by the convex hull of the effects $\dyad{0}$, $\dyad{1}$, $0$ and $\one$.}

    \label{fig:DualConeReachable}
\end{figure}
Consider the case $n=d=2$. A trivial case is the set of measurements that can be obtained from $\P_\alpha=(\alpha\one,(1-\alpha)\one)^\dag$. One verifies that by using a blockwise stochastic matrix with equal columns, any measurement is accessible, as shown in Fig.~\ref{fig:DualConeReachable}$a$.

Let us now analyze the set of accessible measurements from $\P=(P,\one-P)^\dag$ with nonzero projectors $P=P^2$. Writing $\P_Z=(\dyad{0},\dyad{1})^\dag$ in the eigenbasis of $P$, we see that 
one can only obtain measurements with diagonal effects. Therefore, the obtainable set is a fragment of a plane defined by the convex hull of 4 extremal points, depicted in Fig.~\ref{fig:DualConeReachable}$b$.
\begin{equation}
    \{\S\*\P_Z\}_{\S\in{\Delta_{2,2}}^{\times 2}}=\conv\left\{
    \begin{pmatrix}
        \one \\
        0
    \end{pmatrix}\,,
    \begin{pmatrix}
        \dyad{0}\\
        \dyad{1}
    \end{pmatrix}\,,
    \begin{pmatrix}
        \dyad{1}\\
        \dyad{0}
    \end{pmatrix}\,,
    \begin{pmatrix}
        0 \\
        \one
    \end{pmatrix}
    \right\}\,.
\end{equation}
To obtain the full set of measurements, unitary operations $U\in\mathcal{U}(d)$ in the form of $(UPU^\dag,\one-UPU^\dag)^\dag$ are needed in addition to the transformations induced by blockwise product.
\end{example}
Now we will show that the question of what quantum measurements can be obtained from an original quantum measurement $\P$ via the blockwise product of Definition~\ref{def:PositiveProduct}, is in fact related to so-called joint measurability of two POVMs, defined as follows.
\begin{definition}[Jointly measurable]\label{def:compatiblePOVM}
A quantum measurement with $n$ effects $\{P_j\}$ is {\em compatible} or {\em jointly measurable} with a quantum measurement with $m$ effects $\{Q_i\}$, if they are both marginals of a {\em mother} quantum measurement with $n\times m$ effects $\{M_{ij}\}$, namely
\begin{align}
P_j &= \sum_{i=1}^m M_{ij}\quad\text{and}\label{eq:mothercompPj}\\
Q_i &= \sum_{j=1}^n M_{ij}\label{eq:mothercompQi}\,.
\end{align}
\end{definition}
In the special case of a von-Neumann measurement, $\{P_j\}$ and $\{Q_i\}$ are projectors and the corresponding measurements are compatible if and only if the effects commute, $P_jQ_i=Q_iP_j$. Thus the notion of compatibility assesses when two measurements can be implemented simultaneously (see~\cite{OtfriedRev_Joint_2023} for a recent review). 

The problem of deciding measurement compatibility~\cite{Karthik_Joint_2015,Beneduci_Joint_2017,Bluhm_Joint_2018,Jeongwoo_NecSuf_Joint_2019,MitraCompQIns_2022,Buscemi_Joint_2023} can be related to that of post-processing:
In the most general case, it was shown in~\cite{leevi2022incompQI} that two quantum instruments are compatible if and only if one can be obtained by conditional composition with other instruments from each other. The following relation to dynamics through the blockwise product provides a fine-graded version of this result for the case of L\"{u}ders effects:
\begin{prop}\label{thm:ConvertCompat}
Let $\P,\Q\in\Ddn$ define two quantum measurements with effects $\{P_j\}$ and $\{Q_i\}$. These two measurements are compatible if and only if there exists a blockwise stochastic matrix $\S\in{\Ddn}^{\times n}$ mapping $\P$ to $\Q=\S\*\P$. 
\end{prop}
\begin{proof}
We will show first the {\em only if} part, namely that compatibility implies convertibility. Assume $\P$ and $\Q$ define compatible measurements on a system with dimension $d$, and let $\{M_{ij}\}$ be their mother measurement. Define $S_{ij}:=\sqrt{P_j}^{-1}M_{ij}\sqrt{P_j}^{-1}+p_iR_j$ where $\sqrt{P_j}$ is the Hermitian square root of $P_j$ and the inverse is taken on the support $\text{supp}P_j$ of $P_j$. Here $R_j:=\one-\text{supp}P_j$ with $p_i\geq 0$ and $\sum_ip_i=1$. 
Eq.~\eqref{eq:mothercompPj} implies that
\begin{equation}
\sum_{i=1}^n S_{ij}=\sqrt{P_j}^{-1}\sum_{i=1}^nM_{ij}\sqrt{P_j}^{-1}+p_iR_j=\sqrt{P_j}^{-1}P_j\sqrt{P_j}^{-1}+p_iR_j=\one,
\end{equation}
and therefore $\S$ is blockwise stochastic. On the other hand, Eq.~\eqref{eq:mothercompQi} implies that
\begin{equation}
Q_i = \sum_{j=1}^n M_{ij} = \sum_{j=1}^n \sqrt{P_j}S_{ij}\sqrt{P_j}\,,
\end{equation}
and hence $\Q$ can be obtained from $\P$ via $\Q=\S\*\P$.

The {\em if} part becomes now straightforward: let $\S\in{\Ddn}^{\times n}$ be blockwise stochastic such that $\S\*\P=\Q$. Then, the blockwise matrix with entries $M_{ij}=\sqrt{P_j}S_{ij}\sqrt{P_j}$ satisfies Eq.~\eqref{eq:mothercompPj} and~\eqref{eq:mothercompQi} and thus it defines a mother measurement of $\P$ and $\Q$.
\end{proof}
Since compatibility is a symmetric relationship, this implies that all possible transformations between two POVMs $\P\in\Ddn$ and $\Q\in\Ddn$ can be done in both directions, with two blockwise stochastic matrices $\S\in{\Ddn}^{\times n}$ and $\T\in{\Ddn}^{\times n}$ such that $\Q=\S\*\P$ and $\P=\T\*\Q$.

\section{Dynamics induced by blockwise bistochastic matrices}\label{sec:BlockBistoch}
Here we will consider a particular case of the dynamics discussed in the previous section. In the classical case, a stochastic matrix is called bistochastic if both its columns and rows sum to 1. These induce a very special dynamics in the probability simplex and play an important role in physical scenarios~\cite{Nielsen99,Horodecki_ThermoMaj2013,Brand_o_2013}, as they induce growth of entropy in probability vectors. 

Here we study the discrete dynamics of blockwise probability vectors induced by quantum versions of bistochastic matrices~\cite{Benoist_BlockWiseObjs2017,GueriniBProbVAndBBist_2018,Gemma2020}, extending some of the results which are known for $d=1$ to the case $d\geq 2$.
\begin{definition}[Blockwise bistochastic~\cite{Benoist_BlockWiseObjs2017}]\label{def:BlockBist}
Let $\B$ be a square matrix of size $dn\times dn$ composed of $n^2$ blocks $B_{ij}$ of size $d\times d$ each. We call $\B$ \emph{blockwise bistochastic} if 
\begin{enumerate}
    \item its entries $B_{ij}$ are positive semidefinite matrices of size $d\geq 2$
    \item its blockwise columns and rows sum to identity,
    \begin{equation} \sum_iB_{ij}=\sum_{j}B_{ij}=\one_d.
    \end{equation}
\end{enumerate}
\end{definition}
Here we will denote the set of blockwise bistochastic matrices composed of $n^2$ square blocks of size $d$ as $\mathcal{B}_{n,d}$. Similarly as in the definitions of blockwise probability vector and blockwise stochastic matrix, we recover the standard notion of \emph{bistochastic matrix} if $\{B_{ij}\}$ are positive real numbers (setting $d=1$) by selecting any state $\varrho\in\Omega_d$ and considering entry-wise expectation values $\tr(B_{ij}\varrho)$. The simplest non-classical case is $n=2$,
\begin{equation}
    \bm{B}=
    \begin{pmatrix}
    B & \one-B \\
    \one-B & B
    \end{pmatrix},
\end{equation}
where $B$ is any Hermitian matrix of size $d\times d$ satisfying $0\leq B\leq \one$, and hence we have $\mathcal{B}_{2,d}\approx\Delta_{2,d}$. In particular, for $d=n=2$, the set of blockwise bistochastic matrices is inherently described by Figure~\ref{fig:CheeseM=2}.

\subsection{Bistochastic dynamics and partial order between blockwise probability vectors}\label{subsec:DynamicsBistochastic}

Let us first recall the notion of extremal blockwise probability vector $\bm{V_j}=v_j\ot\one$ with a single nonzero component, as defined in equations~\eqref{eq:CartesianVectors} and~\eqref{eq:ExtremalBPV}. In the probability simplex, any probability vector can be obtained by matrix multiplication of a bistochastic matrix times an extremal probability vector $v_i$. In the set of blockwise probability vectors a similar result holds, as shown below.
\begin{obs}\label{obs:AnyAcanbeReachedFromId}
Let $\bm{P}\in\Delta_{n,d}$. For any $1\leq j\leq n$, there exists a blockwise bistochastic matrix $\bm{B}\in\Bdn$ such that 
\begin{equation}
    \bm{B}\*\bm{V_j}=\bm{P}\,.
\end{equation}
\end{obs}
\begin{proof}
Any blockwise probability vector $\bm{P}=(P_1,...,P_n)^\dag$ can be obtained from $\bm{V_1}=(\one,0,...,0)^\dag$ as follows,
\begin{equation}
\begin{pmatrix}
P_1 & P_2 & \cdots & P_d \\
P_2 & P_3 & \cdots & P_1 \\
\vdots & \vdots &  & \vdots \\
P_d & P_1 & \cdots & P_{d-1}
\end{pmatrix}
\*
\begin{pmatrix}
\one \\
0 \\
\vdots \\
0
\end{pmatrix}
=
\begin{pmatrix}
P_1 \\
P_2 \\
\vdots \\
P_d
\end{pmatrix}\,.
\end{equation}
A similar construction does the job for any other vector $\bm{V_j}$.
\end{proof}

Since a blockwise bistochastic matrix is also blockwise stochastic, the blockwise product by a blockwise bistochastic matrix also preserves the set of measurements $\Ddn$. Therefore the product of two blockwise bistochastic matrices is blockwise stochastic, although in general it is not blockwise bistochastic. This is the case only for $n=2$ and arbitrary $d$, and if all entries to be multiplied via the sequential product~\cite{Gudder2001sequential} in the form of $\sqrt{X}Y\sqrt{X}$ commute. Unlike the standard matrix multiplication, 
the blockwise product $\*$ is not associative, and thus
\begin{equation}
    \B\*(\B'\*\P) \neq (\B\*\B')\*\P\,,
\end{equation}
see~\ref{App:PropsProduct}. This means that the discrete dynamics of a quantum measurement $\P\in\Ddn$ by blockwise product by subsequent blockwise bistochastic matrices $\B,\B'\in\Bdn$ in general cannot be obtained by a global matrix $\B\*\B'$.
Due to this fact, the fixed points under blockwise bistochastic dynamics are non-unique (see~\ref{app:FixPoints}).

Standard bistochastic matrices induce special dynamics in the classical probability simplex. If $\b$ is bistochastic and $\pp$ is a probability vector, then the probability vector $\qq=\b\cdot\pp$ is majorized by $\pp$, written $\pp\succ \qq$ (see Theorem II.1.9 in~\cite{Bhatia_MatAn1997}), where the majorization sign $\succ$ means
\begin{equation}\label{eq:MajorClassic}
    \sum_{j=1}^k p_j^\downarrow \geq \sum_{j=1}^k q_j^\downarrow\,\quad\text{for}\quad k=1,...,n,
\end{equation}
where the superindex $\downarrow$ denotes non-increasing order of the components. This induces a partial order in the probability simplex $\Dn$ which plays a key role in quantum information tasks, such as entanglement transformations of pure states~\cite{Nielsen99} or mixed state interconversion under incoherent operations~\cite{Du_MajorCoh2015,Chitambar_MajorCoh2016}. 

Although standard probabilities can always be sorted nonincreasingly, this is not the case for effects of quantum measurements, which are positive semidefinite operators in dimension $d\geq 2$. Therefore, to compare quantum measurements in an analogous way as in the classical case, the following definition needs to be specified.
\begin{figure}[tbp]
\centering
    \includegraphics[scale=1]{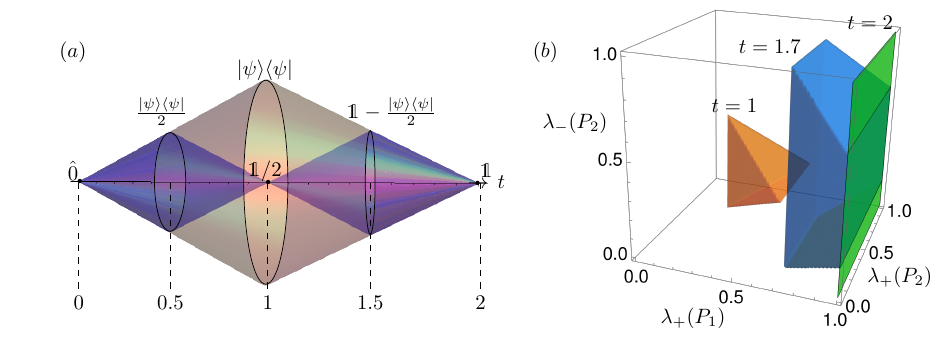}
    \caption{{\bf Visualization of sets of sortable quantum measurements}. ($a$) Three-dimensional dual cone containing the real blockwise vectors $\P=(P,\one-P)^T\in\Delta_{2,2}$ (the imaginary component is omitted) with axial coordinate $\t=\tr(P)$, shown in brown. The blue region contains the set of sortable vectors. ($b$) Three-dimensional cube of vectors $(P_1,P_2,\one-P_1-P_2)^\dag\in\Delta_{3,2}$ with commutative entries, $P_1P_2=P_2P_1$ with a fixed parameter $\tr(P_1)=1$ and a variable parameter $\t=\tr(P_2)$. The axes are the maximum eigenvalue of $P_1$, $\lambda_+(P_1)$, and the maximum and minimum eigenvalues of $P_2$, $\lambda_+(P_2)$ and $\lambda_-(P_2)$. The orange, blue and green regions contain the sets of sortable vectors for $\t=1$, $\t=1.7$ and $\t=2$ respectively.}\label{fig:CheeseComparable}
\end{figure}
\begin{definition}[Sortable quantum measurement]\label{def:SortablePOVM}
A quantum measurement of $n$ effects $P_1$, ..., $P_n$ acting on a $d$-dimensional system is called {\it sortable} if its effects can be sorted in the L\"owner order~\cite{ZhanMatIneq_2002},
\begin{equation}\label{eq:IndepOrder}
    \one_d\geq P_1\geq P_2\geq \dots\geq P_n\geq 0\,.
\end{equation}
\end{definition}
In Fig.~\ref{fig:CheeseComparable} we plot the set of blockwise probability vectors which can be ordered in such way for small system sizes. Sortability of a quantum measurement is physically determined by its output probability distribution with respect to any state of the system, as follows. 
\begin{obs}
A quantum measurement is sortable if and only if regardless of the state $\varrho\in \Omega_d$ of the system, the probability of obtaining an outcome $j$ is larger than the probability of obtaining an outcome $j+1$.
\end{obs}
\noindent This becomes evident since each difference between effects $P_j-P_{j+1}$ is positive semidefinite if and only if for any state $\varrho$ we have
\begin{equation}
    0\leq\tr\big ((P_j-P_{j+1})\varrho\big )=\tr(P_j\varrho)-\tr(P_{j+1}\varrho)\,.
\end{equation}

\subsection{Majorization relation between quantum measurements}\label{subsec:QuantumMajorization}
The Birkhoff - von Neuman theorem~\cite{birkhoff1946tres} establishes that a matrix is bistochastic, if and only if it is a convex combination of permutation matrices, and thus characterizes the set of bistochastic matrices. Looking for a quantum generalization for blockwise bistochastic matrices, it was shown by De les Coves and Netzer~\cite{Gemma2020} that this is not the case, namely the quantum set of blockwise bistochastic matrices is strictly larger.
However, a well-known equivalent formulation of the Birkhoff - Von Neumann theorem, originally stated by Ostrowski~\cite{ostrowski1952quelques} and captured in standard books~\cite{marshall1979inequalities,Bhatia_MatAn1997}, characterizes bistochastic matrices in the following way:

\begin{theorem}[\cite{ostrowski1952quelques}]\label{thm:ClassBVN}
    A stochastic matrix $B$ is bistochastic if and only if any probability vector $p\in\Delta_{n,1}$ majorizes its image, $p\succ Bp$.
\end{theorem}
This formulation allows us to characterize blockwise bistochastic matrices as follows, in a quantum analogy to Theorem~\ref{thm:ClassBVN}.
\begin{restatable}{theorem}{resIndepMaj}[State-independent operator majorization]\label{thm:StateIndep}
Discrete dynamics in the set of blockwise probability vectors satisfies the following equivalence.
\begin{enumerate}
    \item Let $\B\in\Bdn$ be a blockwise bistochastic matrix. Let $\P\in\Ddn$ describe a sortable quantum measurement with image $\Q=\B\*\P\in\Ddn$. Then the vectors $\P$ and $\Q$ satisfy the majorization relation $\P\succ\Q$, written
    \begin{equation}\label{eq:IndepMajorization}
        \sum_{j=1}^kP_j \geq \sum_{j=1}^kQ_j
    \end{equation}
    in the L\"owner order, for all $1\leq k \leq n$, with equality for $k=n$.
    \item Conversely, a matrix $\B\in{\Ddn}^{\times n}$ is blockwise bistochastic if Eq.~\eqref{eq:IndepMajorization} holds for any sortable quantum measurement $\P\in\Ddn$ and for any ordering of the effects of $\Q=\B\*\P\in\Ddn$.
\end{enumerate}
\end{restatable}
\noindent Note that in the classical case $d=1$, the above statement reduces to Theorem~\ref{thm:ClassBVN}. In the quantum case, the set of sortable quantum measurements plays a crucial role in characterizing blockwise bistochastic matrices, although its volume shrinks significantly when the number of effects and dimensions increase.

\begin{proof} We will follow a similar structure as in the proof of Theorem II.1.9 in~\cite{Bhatia_MatAn1997}. 

\noindent (i) Let us first prove that if $\B\in\mathcal{B}_{n,d}$, then Eq.~\eqref{eq:IndepMajorization} holds. For convenience, let us define 
\begin{equation}
    \Theta=\sum_{j=1}^k P_j - Q_j
\end{equation}
and $B_j=\sum_{i=1}^kB_{ij}$, which satisfies
\begin{equation}\label{eq:0Bj}
    \sum_{j=1}^n B_j - k\one = 0\,.
\end{equation}
We want to show that the operator $\Theta$ is positive semidefinite. For that, we introduce Eq.~\eqref{eq:0Bj} and obtain
\begin{equation}\label{eq:ComputeTheta}
\begin{aligned}
\Theta &=  \sum_{j=1}^k P_j - \sum_{j=1}^n\sqrt{P_j}B_j\sqrt{P_j} + 
\sqrt{P_k}\Bigg (\sum_{j=1}^n B_j - k\one\Bigg )\sqrt{P_k} \\
&= \sum_{j=1}^k \sqrt{P_j}(\one-B_j)\sqrt{P_j} - \sqrt{P_k}(\one-B_j)\sqrt{P_k}  \\
&\quad\quad\quad\quad\quad\quad\quad
+ \sum_{j=k+1}^n \sqrt{P_k}B_j\sqrt{P_k} - \sqrt{P_j}B_j\sqrt{P_j} \,.
\end{aligned} 
\end{equation}
We wish to show that the minimal eigenvalue of $\Theta$ is nonnegative. On one hand, as a special case of Weyl's inequalities~\cite{Weyl_Ineq1949}, we have that given two Hermitian matrices $A$ and $B$, the minimal eigenvalue of their sum satisfies the bound
\begin{equation}\label{eq:InequalityLambdaMin}
\lambda_{\min}(A+B)\geq\lambda_{\min}(A)+\lambda_{\min}(B)\,.
\end{equation}
On the other hand, recall the sequential product $A\circ B=\sqrt{A}B\sqrt{A}$ defined in~\cite{Gudder2001sequential}. Since by assumption $P_a-P_{b\geq a}\geq 0$, due to Observation~\ref{obs:SeqProdEqualSpectra} the spectrum of $(P_a-P_{b\geq a})\circ(\one-B_c)$ is the same as the spectrum of $(\one-B_c)\circ(P_a-P_{b\geq a})$. Combining these two facts, we have
\begin{equation}
\lambda_{\min}(\Theta) \geq \sum_{j=1}^k \lambda_{\min}\big ((\one - B_j)\circ (P_j-P_k) \big ) + \sum_{j=k+1}^n \lambda_{\min}\big ( B_j\circ (P_k-P_j)\big ).
\end{equation}
Since the sequential product preserves positivity, all terms above are nonnegative and we have $\lambda_{\min}(\Theta)\geq 0$, which implies that $\Theta$ is positive semidefinite. This proves item (i).

(ii) Now we will show the converse direction, namely that if Eq.~\eqref{eq:IndepMajorization} holds for all sortable vectors $\P\in\Ddn$ satisfying~\eqref{eq:IndepOrder}, then $\B$ is blockwise bistochastic. For that we recall that Eq.~\eqref{eq:IndepMajorization} needs to be fulfilled in particular for $\bm{V_j}=(0,...,0,\one,0,...,0)$ with $\one$ at the $j$-th position and 0 elsewhere. This implies that
\begin{equation}
    \sum_{i=1}^kB_{ij} \leq \one
\end{equation}
with equality for $k=n$, which implies that $\B$ is blockwise stochastic. Eq.~\eqref{eq:IndepMajorization} needs to be fulfilled also for $\bm{P_u}=(\one,...,\one)/n$. This implies that
\begin{equation}\label{eq:MajForPu}
\frac{1}{n}\sum_{i=1}^k\bigg (\sum_{j=1}^n B_{ij}\bigg ) \leq \frac{k}{n}\one\,,
\end{equation}
with equality for $k=n$. Now assume that for some $k$, Eq.~\eqref{eq:MajForPu} holds with a strict inequality. Then it cannot hold with equality for $k=n$, which is a contradiction. Therefore, the blockwise rows of $\B$ sum exactly to identity and thus $\B$ is blockwise bistochastic. 
\end{proof}

The theorem above can physically be interpreted as in the following Corollary. 
\begin{restatable}{cor}{rescorStateInd}\label{cor:InterpMatrixMaj}
Suppose a quantum system is prepared in a state $\varrho$. A sortable measurement $\P$ is performed with outcome probability distribution $\pp$, and then a second measurement $\S_j$ is performed depending on the outcome $j$, with outcome probability distribution $\qq$ (see Fig.~\ref{fig:2MeasurementsGraphicRep}). The following statements are equivalent.

\begin{enumerate}
    \item For any quantum state $\varrho$, the output probability vector $\qq$ can be obtained with classical post-processing from the vector $\pp$ as
    \begin{equation}
        \qq = B\pp\,,
    \end{equation}
    where $B$ is a standard bistochastic matrix of order $n$.
    \item For any initial state $\varrho$, the first and second probability distributions $\pp$ and $\qq$ satisfy the majorization relation $\pp\succ\qq$.
    \item The collection of the $i$-th effect of each of the possible second measurements $\S_j$, $\{S_{i1},...,S_{in}\}$, defines a quantum measurement.
\end{enumerate}
\end{restatable}
\begin{proof}
$(i)\iff(ii)$ It is a standard result in classical probability theory that there exists a bistochastic matrix $B$ such that $\qq=B\pp$, if and only if the majorization relation $\pp\succ\qq$ holds~\cite{Bhatia_MatAn1997}.

\noindent $(ii)\iff(iii)$ The majorization relation $\pp\succ \qq$ can be written as
\begin{equation}
    \sum_{j=1}^k \tr(P_j\varrho)\geq \sum_{j=1}^k \tr(Q_j\varrho)\,,
\end{equation}
which holds for any state $\varrho$ if and only if the majorization relation $\P\succ\Q$ of Eq.~\eqref{eq:IndepMajorization} holds, where $\P$ is sortable by assumption. By Theorem~\ref{thm:StateIndep}, this holds for any state $\varrho$ and measurement $\P$ if and only if $\S=(\S_1,...,\S_n)$ is blockwise bistochastic, which by definition means that the sets $\{S_{1j},...,S_{nj}\}$ and $\{S_{i1},...,S_{in}\}$ are quantum measurements for all $i$ and $j$. Whereas $\{S_{1j},...,S_{nj}\}$ is a quantum measurement for all $j$ by the stochasticity assumption, the fact that $\{S_{i1},...,S_{in}\}$ is a measurement for all $i$ is exactly the additional condition for bistochasticity.
\end{proof}

By arguing in an analogous way as in Theorem~\ref{thm:StateIndep}, one can establish a state-dependent majorization relation under blockwise bistochastic dynamics.
\begin{restatable}{theorem}{resStateDep}[State-dependent operator majorization]\label{thm:DepMajor}
Let $\P\in\Ddn$, $\B\in\mathcal{B}_{n,d}$ and $\Q=\B\*\P$. Given a quantum state $\varrho\in \Omega_d$, assume that $\P$ can be ordered such that for any $s\geq r$ the minimal eigenvalues $\lambda_{\min}$ satisfy the inequality
\begin{equation}\label{eq:AssumptionLarger}
    \lambda_{\min}\Big (\sqrt{P_r}\varrho\sqrt{P_r}\Big )\geq\lambda_{\min}\Big (\sqrt{P_{s}}\varrho\sqrt{P_{s}}\Big )\geq 0.
\end{equation}
Then the classical probability vectors with components $\pp_j=\tr(P_j\varrho)$ and $\qq_i=\tr(Q_i\varrho)$ satisfy the standard majorization relation $\pp\succ\qq$.
\end{restatable}
\noindent Note that Theorem~\ref{thm:DepMajor} holds in particular if $\sqrt{P_i}\varrho\sqrt{P_i}\geq\sqrt{P_{j\leq i}}\varrho\sqrt{P_{j\leq i}}$, which is a state-dependent analogue to the condition for Theorem~\ref{thm:StateIndep}.

\begin{proof}
Analogously as in Theorem~\ref{thm:StateIndep}, we want to show that the real number
\begin{equation}
\theta = \sum_{j=1}^k
\tr(P_j\varrho)-\tr(Q_j\varrho)\,
\end{equation}
is nonnegative. In analogy to Eq.~\eqref{eq:ComputeTheta} concerning the operator $\Theta$, here we obtain the following expression for the scalar $\theta$,
\begin{equation}
\begin{aligned}
\theta  &= \tr\Bigg ( \sum_{j=1}^k (\one-B_j)\left(\sqrt{P_j}\varrho\sqrt{P_j} - \sqrt{P_k}\varrho\sqrt{P_k}\right) \\
&\quad\quad\quad\quad\quad\quad\quad
+ \sum_{j=k+1}^n B_j \Big ( \sqrt{P_k}\varrho\sqrt{P_k}-\sqrt{P_j}\varrho\sqrt{P_j} \Big )\Bigg )\,,
\end{aligned} 
\end{equation}
where $B_j=\sum_{i=1}^k B_{ij}$. To proceed we need the property that given an Hermitian matrix $X$ and a positive semidefinite matrix $Y$, we have
\begin{equation}\label{eq:IneqPosHerm}
\tr(XY)=\tr(\sqrt{Y}X\sqrt{Y})\geq\tr(\sqrt{Y}\one\lambda_{\min}(X)\sqrt{Y})=\tr(Y)\lambda_{\min}(X)\,.
\end{equation}
This property can be applied as follows. Since $\B$ is blockwise bistochastic, $(\one-B_j)$ is positive semidefinite. Since $\sqrt{P_i}\varrho\sqrt{P_i}$ is positive semidefinite for any $i$, the operator $\sqrt{P_k}\varrho\sqrt{P_k}-\sqrt{P_j}\varrho\sqrt{P_j}$ is Hermitian and therefore Eq.~\eqref{eq:IneqPosHerm} applies. Thus we have
\begin{equation}
\begin{aligned}
\theta&\geq  \sum_{j=1}^k \tr\big (\one-B_j\big )\lambda_{\min}\left(\sqrt{P_j}\varrho\sqrt{P_j} - \sqrt{P_k}\varrho\sqrt{P_k}\right) \\
&\quad\quad\quad\quad\quad\quad\quad
+ \sum_{j=k+1}^n \tr(B_j) \lambda_{\min}\Big ( \sqrt{P_k}\varrho\sqrt{P_k}-\sqrt{P_j}\varrho\sqrt{P_j} \Big ).
\end{aligned}
\end{equation}
Using Eq.~\eqref{eq:InequalityLambdaMin} again we can further bound $\theta$ from below by splitting the terms as
\begin{equation}
\begin{aligned}
\theta&\geq  \sum_{j=1}^k \tr\big (\one-B_j\big )\bigg (\lambda_{\min}\left(\sqrt{P_j}\varrho\sqrt{P_j}\Big ) - \lambda_{\min}\Big (\sqrt{P_k}\varrho\sqrt{P_k}\right)\bigg ) \\
&\quad\quad\quad\quad\quad\quad\quad
+ \sum_{j=k+1}^n \tr(B_j) \bigg (\lambda_{\min}\Big ( \sqrt{P_k}\varrho\sqrt{P_k}\Big )-\lambda_{\min}\Big (\sqrt{P_j}\varrho\sqrt{P_j} \Big )\bigg ).
\end{aligned}
\end{equation}
By assumption $\P$ is ordered according to Eq.~\eqref{eq:AssumptionLarger} and therefore all factors are nonnegative.
\end{proof}

\section{Resource theory of quantum measurements}\label{sec:ResThryPOVMs}
Resource theories are abstract frameworks which characterize the possible transformations within physical scenarios. Any resource theory has three main ingredients~\cite{GourBrandaoResTheories2015,ChitambarGourQRTheos_2019}: (i) a set of free elements, (ii) a set of transformations that leave invariant the free elements, and (iii) a set of resourceful elements. It is often useful to introduce two extra ingredients: (iv) the subset of maximally resourceful elements, from which any other element can be obtained with free operations, and (v) a monotone, which is a quantity associated to each element that cannot increase under free operations, and thus it quantifies how resourceful a given element is. 

Recently, resource theories became a useful tool to analyze the set of POVMs~\cite{oszmaniecOpRelResTQM_2019,Guff_ResTheorPOVMs2021,Tendick_ResTheorPOVMs2022,buscemiCompSharpResTQM_2024}. Here we show that the framework introduced in this work allows to establish an alternative resource theory, by making use of generalized notions of entropy and majorization. Operationally, the standard resource theory in the probability simplex $\Delta_{n}$ is obtained by considering the expectation value of quantum effects with respect to the state of the quantum system in hand.

To describe the discrete dynamics within $\Delta_{n}$, whose elements are column probability vectors, the resource theory of majorization is well established. The free elements (i) are the uniform vectors with equal components, the free operations (ii) are bistochastic matrices equipped with the matrix product on probability vectors, the set of resourceful elements (iii) are the probability vectors which are not uniform, the extremal vectors $\{{v_i}\}$ defined in~\eqref{eq:CartesianVectors} (iv) are maximally resourceful, and the monotones (v) are the so-called \emph{cumulative probability distributions} of a probability vector sorted in non-increasing order, $\mu_k=\sum_{j=1}^{k\leq n}{p_j}^\downarrow$.

The quantum case has a more involved structure, in several senses: the extremal points lie in a continuous high-dimensional hypersurface (see~\cite{OszmaniecExtrPOVMs_2016,JenExtrmGenPOVMs_2013}); the elements which remain invariant by a given free operation are not straightforward to study (see~\ref{app:FixPoints}); and several functionals might be chosen to describe dynamics (see Fig.~\ref{fig:DynamicsCheeseSmNoise} and~\ref{app:NonlinearMonotone}). Nonetheless, a resource theory can be established naturally as follows.  

\begin{enumerate}
    \item[(i-ii)] The set of blockwise bistochastic matrices $\mathcal{B}_{n,d}$ describes the free operations and the vector $\bm{P_u}=(\one_d/n,...,\one_d/n)^\dag\in\Ddn$ describes the free element for each $n$ and $d$, since for any $\bm{B}\in\mathcal{B}_{n,d}$ it holds that
    \begin{equation}
        \bm{B}*\bm{P_u}=\bm{P_u}\, ,
    \end{equation}
    and therefore with $\mathcal{B}_{n,d}$ one cannot reach any point outside of $\{\bm{P_u}\}$. 
    \item[(iii)] From any point $\bm{P}\in\Delta_{n,d}$ one can reach a free element $\bm{P_u}$ by choosing $\bm{B_u}\in\mathcal{B}_{n,d}$ defined component wise by $(B_u)_{ij}=\one_d/n$ as
    \begin{equation}
        \bm{B_u}*\bm{P}=\bm{P_u}\,.
    \end{equation}
    Thus the set of resourceful states can be identified with $\Delta_{n,d}\setminus \{\bm{P_u}\}$.
    \item[(iv)] The maximally resourceful elements with respect to the blockwise product are those blockwise probability vectors $\bm{P}\in\Delta_{n,d}$ with the property that for any $\bm{Q}\in\Delta_{n,d}$, there exists $\B\in\mathcal{B}_{n,d}$ such that
    \begin{equation}
        \bm{Q}=\B\*\bm{P}\,.
    \end{equation}
    According to Observation~\ref{obs:AnyAcanbeReachedFromId}, the vectors $\{\bm{V^{(i)}}\}$ defined in Eq.~\eqref{eq:CartesianVectors} are maximally resourceful. This stands in contraposition to the resource theory introduced in~\cite{buscemiCompSharpResTQM_2024}, where $\{\bm{V^{(i)}}\}$ are precisely the free elements.
    \item[(v)] Theorem~\ref{thm:StateIndep} provides us with a natural state-dependent monotone $E_\varrho$, for sortable quantum measurements, which is defined as
    \begin{equation}\label{eq:MonotoneSDep}
        E_\varrho(\P)=\sum_{j=1}^n\log\big (\tr(P_j\varrho)\big )+n\log n\,,
    \end{equation}
    and a natural state-independent monotone,
    \begin{equation}\label{eq:MonotoneSInd}
        E=\min_{\varrho\in \Omega_d}\big (E_\varrho(\P)\big )\,.
    \end{equation}
    This is shown in the following lemma.
\end{enumerate}
\begin{lemma}
    Let $\P\in\Ddn$ be a sortable blockwise probability vector and let $\B\in\Bdn$ be blockwise bistochastic. Let $\Q=\B\*\P$. For any quantum state $\varrho\in \Omega_d$, the entropy $E_\varrho$ satisfies
    \begin{equation}\label{eq:SIEntropy}
        E_\varrho(\P)\geq E_\varrho(\Q)\,.
    \end{equation}
    Moreover, one has $E_\varrho(\P_u)=0$ for any quantum state $\varrho$.
\end{lemma}
\begin{proof}
    By Theorem~\ref{thm:StateIndep} and Corollary~\ref{cor:InterpMatrixMaj}, the vector $\pp=(\tr(P_1\varrho),...,\tr(P_n\varrho))$ majorizes the vector $\qq=(\tr(Q_1\varrho),...,\tr(Q_n\varrho))$, written $\pp\succ\qq$ -- see Eq.~\eqref{eq:MajorClassic}. This implies~\cite{Bhatia_MatAn1997} that the standard entropy of $\pp$ is smaller than the entropy of $\qq$, which implies that $E_\varrho(\P)\geq E_\varrho(\Q)$ for any state $\varrho$. 
    Therefore, the quantity $E_\varrho$ defined in Eq.~\eqref{eq:MonotoneSDep} cannot increase under blockwise product by a blockwise bistochastic matrix.

    It is straightforward to see that $E_\varrho(\P_u)= -n\log n + n\log n = 0$ for the (sortable) free measurement $\P_u=(\one,...,\one)^T/n$ defined in~\eqref{eq:UniformBPV}.  Since the results hold for any state $\varrho$, they also hold for the tightest case of Eq.~\eqref{eq:SIEntropy} obtained by minimizing the monotone over $\varrho$.
\end{proof}
A possible state-independent majorization relation for non-sortable quantum measurements based on non-linear quantities is analyzed in~\ref{app:NonlinearMonotone}.

\section{Concluding remarks}\label{sec:Conclusions}
We introduced a framework to analyze dynamics in the set $\Ddn$ of quantum measurements, such that the outcome probabilities can be simulated by concatenation with a conditional measurement. For that we consider a blockwise product with blockwise stochastic matrices, which define collections of quantum measurements.
This approach, inspired by a work of Gudder~\cite{Gudder08}, can be seen as a generalization of discrete dynamics of probability distributions induced by stochastic matrices. 
Unlike in the classical case, transformations induced in this way are possible only when the input and output measurements are compatible. 

Special dynamics is obtained if considering blockwise bistochastic matrices, where all columns and rows define a quantum measurement. We have seen that, as in the classical case, these induce an operator majorization relation between quantum measurements. 
This generalization of central results in standard majorization theory characterizes the set of blockwise bistochastic matrices~\cite{Benoist_BlockWiseObjs2017,Gemma2020} and
can be understood in terms of  classical postprocessing of a sequence of quantum measurements. 
The proposed approach allows to establish a resource theory of quantum measurements, which can be seen as a noncommutative version of the theory of majorization in the probability simplex. These results are summarized in Table~\ref{tab:Connections}.

Several questions and possible applications remain open to further study. What extra operations are needed to fully describe transformations between incompatible measurements, is at the moment unknown. Moreover, further studies are needed to understand the structure of the set of blockwise probability vectors $\Ddn$ and relevant subsets, as well as the structure of the set of blockwise bistochastic matrices $\Bdn$. On the fundamental side, the systematic framework for measurement interconversion introduced here is suitable to study contextuality in sequential measurements~\cite{OtfriedContSeq_2010,XiaoExpContext_2016,WangExpContextIons_2022}. For quantum protocols, one could engineer an optimal algorithmic sequence of conditional measurements to obtain particularly interesting measurements, such as informationally complete ones, with certain precision. 

\begin{table}[ht]
\begin{tabular}{p{220pt} p{220pt}}
\textbf{$n$-point probability simplex $\Delta_n$} & \textbf{Set of quantum measurements $\Delta_{n,d}$} \\
\midrule
\\
Probability vector
\vspace{5pt}
\newline
$\,\,p=(p_1,...,p_n)^T\in\Dn\quad\quad\quad\quad
\newline
p_j\geq 0,\quad \sum_{j=1}^n p_j=1$
&
Blockwise probability vector (POVM)~\cite{GueriniBProbVAndBBist_2018}
\vspace{5pt}
\newline
$\,\,\bm{P}=(P_1,...,P_n)^\dag\in\Ddn\quad\quad\quad\quad
\newline
 P_j\geq 0,\quad \sum_{j=1}^n P_j=\one_d$
\\\\
Stochastic matrix
\vspace{5pt}
\newline
$\,\,S=(s_1,...,s_n)\in{\Dn}^{\times n}
\newline
s_j=(s_{1j},...,s_{nj})\in\Delta_n$ 
 &
Blockwise stochastic matrix
\vspace{5pt}
\newline
$\,\,\bm{S}=(\bm{S_1},...,\bm{S_n})\in{\Ddn}^{\times n}
\newline
\bm{S_j}=(S_{1j},...,S_{nj})\in\Delta_{n,d}$
\\\\
Transformations within $\Delta_n$
\vspace{5pt}
\newline
$\,\,Sp=q\quad\quad\quad\quad
\newline
q_i=\sum_{j=1}^n s_{ij}p_j$
 &
Transformations within $\Delta_{n,d}$
\vspace{5pt}
\newline
$\,\,\bm{S}*\bm{P}=\bm{Q}\quad\quad\quad\quad
\newline Q_i=\sum_{j=1}^n \sqrt{P_j}S_{ij}\sqrt{P_j}$
\\\\
Allowed transformations
\vspace{5pt}
\newline
$\,\,p\underset{S}{\rightarrow} q\in\Delta_n$ always possible.
&
Allowed transformations
\vspace{5pt}
\newline
$\,\,\bm{P}\underset{\bm{S}}{\rightarrow}\bm{Q}\in\Delta_{n,d}$ $\iff$ Jointly measurable
\\\\
Bistochastic
\vspace{5pt}
\newline
$\,\,B=(b_{ij})\quad\quad\quad\quad
\newline
b_{ij}\geq 0;\quad\sum_ib_{ij}=\sum_jb_{ij}=1$
&
Blockwise bistochastic~\cite{Benoist_BlockWiseObjs2017,Gemma2020}
\vspace{5pt}
\newline
$\,\,\bm{B}=(B_{ij})\quad\quad\quad\quad
\newline
B_{ij}\geq 0;\quad\sum_{i=1}^n B_{ij}=\sum_{j=1}^n B_{ij}=\one_d$
\\\\
Sortability
\vspace{5pt}
\newline
$\,\,$Nonincreasing order,\quad\quad\quad\quad
\newline
 $1\geq p_1\geq...\geq p_n\geq 0$
&
Sortability
\vspace{5pt}
\newline
$\,\,$Sortable subset,\quad\quad\quad\quad
\newline
 $\one\geq P_1\geq...\geq P_n\geq 0$
\\\\
Majorization for all $p\in\Delta_n$
\vspace{5pt}
\newline
$\,\,p\succ q=Bp$:\quad\quad\quad\quad
\newline
 $\quad\sum_{j=1}^k p_j\geq\sum_{j=1}^k q_j$
&
Majorization for $\bm{P}\in$ sortable $\subset\Delta_{n,d}$
\vspace{5pt}
\newline
$\,\,\bm{P}\succ \bm{Q}=\bm{B}*\bm{P}$:\quad\quad\quad\quad
\newline
 $\quad\sum_{j=1}^k P_j\geq\sum_{j=1}^k Q_j$
\end{tabular}
\caption{Comparison of the main elements of classical probability theory and their generalization for the set of quantum measurements. Since in the quantum case one considers positive operators instead of nonnegative real numbers, some of the connections established in the classical case need additional structure to hold in the quantum case.}
\label{tab:Connections}
\end{table}

\newpage

\section*{Acknowledgements}
It is a pleasure to thank Moisés Bermejo Morán, Dagmar Bru\textbeta,  Shmuel Friedland, Wojciech Górecki, Felix Huber, Kamil Korzekwa, Leevi Lepp\"{a}j\"{a}rvi, Oliver Reardon-Smith, Fereshte Shahbeigi and Ryuji Takagi for fruitful discussions. 
We thank an anonymous referee for detailed comments that helped to improve the presentation of this work. 
Financial support by the Foundation for Polish Science through the Team-Net Project No. POIR.04.04.00-00-17C1/18-00 and 
by NCN QuantERA
Project No. 2021/03/Y/ST2/00193 is gratefully acknowledged.

\newpage
\appendix

\section{Properties of the blockwise product and its dual version}\label{App:PropsProduct}
Here we will discuss the properties of the blockwise product $*$ introduced in Definition~\ref{def:PositiveProduct}, which mathematically has a natural dual version $*^\dag$ upon conjugate transposition. 
\subsection{Unitary invariance of the sequential product}
A key element of the blockwise product of Definition~\ref{def:PositiveProduct} is the sequential product between two positive semidefinite matrices $A\circ B=\sqrt{A}B\sqrt{A}$ introduced in~\cite{Gudder2001sequential}, where $\sqrt{A}$ is the unique positive square root of $A$. This product satisfies the following property, which we use to prove Theorem~\ref{thm:StateIndep}.
\begin{obs}\label{obs:SeqProdEqualSpectra}
The sequential products $A\circ B$ and $B\circ A$ have the same spectrum.
\end{obs}
\noindent This should not be confused with the fact that $A\geq A\circ B$ but $A\not\geq B\circ A$, pointed out in~\cite{Gudder2001sequential}.
\begin{proof}
Defining $X=\sqrt{A}\sqrt{B}$, we have $A\circ B = XX^\dag$ and $B\circ A = X^\dag X$. By computing the singular value decomposition of the matrix $X$, $X=UDV$ where $U$ and $V$ are unitary and $D$ is diagonal, we have
\begin{equation}
    XX^\dag = UDVV^\dag D^*U^\dag = UDD^*U^\dag
\end{equation}
and
\begin{equation}
    X^\dag X = V^\dag D^*U^\dag UDV = V^\dag D^*DV\,.
\end{equation}
Therefore, we have $A\circ B = UV(B\circ A)V^\dag U^\dag$. Since $UV$ is unitary and $D^*D=DD^*$ (as $D$ is diagonal), both products have the same spectra.
\end{proof}

\subsection{The blockwise product and its dual definition}
Here we will analyze a dual version of the blockwise product $*$, which in the classical case cannot be identified. To this end, we will use the fact that given two blockwise probability vectors $\bm{P},\bm{Q}\in\Ddn$, their adjoints $\bm{P}^\dag$ and $\bm{Q}^\dag$ are row vectors. We recall that this is because the blocks $P_j$ and $Q_j$ are Hermitian and thus the adjoint transposes the position the block components, while leaving the blocks invariant. Similarly, the adjoint $\bm{S}^\dag$ defines a blockwise matrix with positive semidefinite entries which sum to identity row-wise.  

Let us now take a closer look at the blockwise product 
$*$, in comparison to its classical version. While in the classical case one obtains the same result choosing $q=Sp$ or $q^T=p^TS^T$ up to transposition, in the quantum regime one obtains different results with $\bm{Q}=\bm{S}\*\bm{P}$ and $\bm{Q}^{\dag}=\bm{P}^\dag\*\bm{S}^\dag$. This is illustrated in Proposition~\ref{prop:DynamicsQsimplexNoGo} below, which is a dual analog to Propositions~\ref{prop:PreserveStochastic} and~\ref{prop:Interconv_Q_Simplex}. For completeness, we define the following dual version of the blockwise product introduced in Definition~\ref{def:PositiveProduct},
\begin{equation}\label{eq:defPosProdDual}
    \bm{A}\,\*^{\dag}\,\bm{B}:=(\bm{B}^\dag\*\bm{A}^\dag)^\dag\,,
\end{equation}
which acts as a sequential product in the form of $\sqrt{A_{ij}}B_{jk}\sqrt{A_{ij}}$ as demonstrated below in Eq.~\eqref{eq:SP=Q}.
\begin{prop}\label{prop:DynamicsQsimplexNoGo}
The following properties hold in the set of measurements $\Delta_{n,d}$ for the dual blockwise product $*^\dag$:
\begin{enumerate}
    \item Given a blockwise probability vector $\bm{P}\in\Delta_{n,d}$ and a blockwise stochastic matrix $\bm{S}\in{\Delta_{n,d}}^{\times n}$, let $\bm{Q}$ be defined as $\bm{Q}=\bm{S}\,\*^{\dag}\,\bm{P}$. Then, $\bm{Q}$ is \emph{not necessarily} a blockwise probability vector.
    \item Let $\bm{P}$ and $\bm{Q}$ be two blockwise probability vectors, $\bm{P},\bm{Q}\in\Delta_{n,d}$. There exists a blockwise stochastic matrix $\bm{S}\in{\Delta_{n,d}}^{\times n}$ such that $\bm{Q}=\bm{S}\,\*^{\dag}\,\bm{P}$.
\end{enumerate}
\end{prop}
\begin{proof}
A counterexample for {\it (i)} can be found by defining the blockwise stochastic matrix $\S$ and blockwise probability vector $\P$ as
\begin{equation}
\S=
\begin{pmatrix}
    P_{Z+} & P_{X+} \\
    P_{Z-} & P_{X-}
\end{pmatrix}
\quad\text{and}\quad
\P=
\begin{pmatrix}
    \alpha P_{X+} + \beta P_{X-} \\
    (1-\alpha) P_{X+} + (1-\beta) P_{X-}
\end{pmatrix}\,,
\end{equation}
where the rank-1 projectors $P_{Z+}=\dyad{0}$, $P_{Z-}=\dyad{1}$, $P_{X+}=\dyad{+}$ and $P_{X-}=\dyad{-}$ diagonalize the Pauli matrices $\sigma_Z$ and $\sigma_X$ with eigenvalue $+1$ and $-1$ respectively, and $0\leq \alpha,\beta\leq 1$. Given $\Q=\S\,\*^{\dag}\,\P$, the sum of the components of $\Q$ reads in general
\begin{equation}
\frac{\alpha+\beta}{2}\one + (1-\alpha)P_{X+} + (1-\beta)P_{X-} \neq\one
\end{equation}
for $0<|\alpha-\beta|< 1$.

To show item (ii), similarly to the classical case one can choose $\bm{S}$ to be
\begin{equation}\label{eq:TrivialExampleQspanAll}
    \bm{S}=
    \begin{pmatrix}
    \bm{Q} & \dots & \bm{Q}
    \end{pmatrix}
    =
    \begin{pmatrix}
    Q_1 & \dots & Q_1 \\
    \vdots & \ddots & \\
    Q_n &  & Q_n
    \end{pmatrix}\,.
\end{equation}
Then one has
\begin{equation}\label{eq:SP=Q}
    \S\*^\dag\P=
    \begin{pmatrix}
    \sum_j\sqrt{Q_1}P_j\sqrt{Q_1} \\
    \vdots\\
    \sum_j\sqrt{Q_n}P_j\sqrt{Q_n}
    \end{pmatrix}
    =\bm{Q}\,,
\end{equation}
as the sum over $j$ factorizes.

\end{proof}
We have seen that the product $\*$ preserves the set $\Ddn$, whereas the product $\,\*^{\dag}\,$ does not. Moreover, the product $\*$ can be interpreted causally in the sense that the measurement which is performed first determines the second measurement, whereas the analogous interpretation of the product $\,\*^{\dag}\,$ is not causal: even though the outcome $j$ of a measurement $\P$ determines a second measurement $\S_j$, the measurement $\S_j$ is performed before $\P$ in the sense that all terms inside $\S\,\*^{\dag}\,\P$ are of the form $\sqrt{S_{ij}}P_j\sqrt{S_{ij}}$. Although from a mathematical perspective one can choose both the product $*$ of Eq.~\eqref{eq:DefPosProd} and its dual $\*^\dag$ in~\eqref{eq:defPosProdDual}, the first option has a direct physical interpretation as it describes discrete dynamics in the set $\Ddn$ induced by concatenation of quantum measurements. This is the reason why we consider this option in the main text of this work.

\subsection{Algebraic properties of $*$ and $*^{\dag}$}\label{App:Properties*}

The blockwise product $*$ defined in Eq.~\eqref{eq:DefPosProd} and its dual $\,\*^{\dag}\,$ defined in Eq.~\eqref{eq:defPosProdDual} have the following properties, where $\neq$ denotes that equality does not hold in general.
\begin{enumerate}
\item Non-commutativity,
\begin{equation}\label{eq:NotCommutative}
\begin{aligned}
    \bm{A}\,\*^{\dag}\,\bm{B} &\neq \bm{B}\,\*^{\dag}\,\bm{A}\,; \\
    \bm{A}\*\bm{B} &\neq \bm{B}\*\bm{A}
\end{aligned}
\end{equation}
\item Partial distributivity, in the sense that
\begin{equation}\label{eq:PartiallyDistributive}
\begin{aligned}
    \bm{A}\,\*^{\dag}\,(\bm{B}+\bm{C}) &= \bm{A}\,\*^{\dag}\,\bm{B} + \bm{A}\,\*^{\dag}\,\bm{C}\quad\text{but} \\
    (\bm{A}+\bm{B})\,\*^{\dag}\,\bm{C} &\neq \bm{A}\,\*^{\dag}\,\bm{C}+\bm{B}*\bm{C}\,; \\
    (\bm{A}+\bm{B})\*\bm{C} &= \bm{A}\*\bm{C}+\bm{B}*\bm{C}\quad\text{but} \\
    \bm{A}\*(\bm{B}+\bm{C}) &\neq \bm{A}\*\bm{B} + \bm{A}\*\bm{C}\,.
\end{aligned}
\end{equation}
\item Non-associativity,
\begin{equation}\label{eq:NotAssociative}
\begin{aligned}
    (\bm{A}\,\*^{\dag}\,\bm{B})\,\*^{\dag}\,\bm{C} &\neq \bm{A}\,\*^{\dag}\,(\bm{B}\,\*^{\dag}\,\bm{C})\,; \\
    (\bm{A}\*\bm{B})\*\bm{C} &\neq \bm{A}\*(\bm{B}\*\bm{C})\,.
\end{aligned}
\end{equation}
\item Unlike the standard product of matrices behaves under the Hermitian conjugate operation $\dag$ as $(A\cdot B)^\dag=B^\dag\cdot A^\dag$, now we have that
\begin{equation}
\begin{aligned}
    (\bm{A}\,\*^{\dag}\,\bm{B})^\dag\neq \bm{B}^\dag \,\*^{\dag}\, \bm{A}^\dag\,; \\
    (\bm{A}\*\bm{B})^\dag\neq \bm{B}^\dag \* \bm{A}^\dag\,.
\end{aligned}
\end{equation}
\end{enumerate}

\section{Fixed points under blockwise bistochastic dynamics}\label{app:FixPoints}
Given a blockwise bistochastic matrix $\bm{B}\in\mathcal{B}_{n,d}$, here we study which quantum measurements in $\Delta_{n,d}$ remain invariant under the dynamics induced by $\bm{B}$. We define these measurements as follows.
\begin{definition}[Fixed points]
Given $\bm{B}\in\mathcal{B}_{n,d}$ and $\bm{P}\in\Delta_{n,d}$, $\bm{P}$ is a \emph{fixed point} of $\bm{B}$ with respect to the operation $*$ ($*^\dag$) if
\begin{equation}\label{eq:DefFixPoint}
\begin{aligned}
    \bm{B}*\bm{P} &=\bm{P} \\
    (\bm{B}*^\dag\P &=\bm{P})\,.
\end{aligned}
\end{equation}
If no specification about $\*$ or $\*^\dag$ is given, this denotes that $\P$ is a fixed point of $\B$ with respect to both products.
\end{definition}
Let us focus on the simplest case, which is $n=2$ and arbitrary $d$. It will be convenient to consider the first entries $B$ and $P$ of a blockwise bistochastic matrix and a blockwise probability vector of the form
\begin{equation}\label{eq:BandP2x2}
    \B=\begin{pmatrix}
        B & \one-B \\
        \one-B & B
    \end{pmatrix}
    \quad\text{and}\quad
    \begin{pmatrix}
        P \\
        \one-P
    \end{pmatrix}\,.
\end{equation}
\begin{obs}\label{obs:FixPointsAnsatz}
Let $\bm{B}\in\mathcal{B}_{2,d}$ be a blockwise bistochastic matrix defined by an entry $B$ and let $\bm{P}\in\Delta_{2,d}$ be a blockwise probability vector defined by an entry $P$. Let $B$ and $P$ have the following form,
\begin{equation}\label{eq:FixPointsCommuteAnsatz}
    B=\one+X\quad\text{and}\quad P=\frac{\one+Y}{2}\,,
\end{equation}
where $X$ and $Y$ have support in mutually orthogonal subspaces, $X\cdot Y=0$. Then $P$ is a fixed point of $B$ with respect to both $*$ and $*^\dag$.
\end{obs}
\begin{proof}
We can write equation~\eqref{eq:DefFixPoint} in terms of commutators $[R,S]=R\cdot S-S\cdot R$ as
\begin{equation}
\begin{aligned}
    (\one-2P)(\one-B)&=\sqrt{\one-B}\Big [P,\sqrt{\one-B}\Big ]-\sqrt{B}\Big [P,\sqrt{B}\Big ]\quad\text{and} \\
    (\one-2P)(\one-B)&=\sqrt{\one-P}\Big [B,\sqrt{\one-P}\Big ]-\sqrt{P}\Big [B,\sqrt{P}\Big ] \,
\end{aligned}
\end{equation}
for the products  $\*^\dag$ and $\*$ respectively. Note that in both cases, if $[B,P]=0$ the right hand side vanishes and one has
\begin{equation}\label{eq:FixPointCommuteEquation}
    \one - 2P - B +2P\cdot B = 0\,.
\end{equation}
In particular, by inspection we see that if $B$ and $P$ take the form of equation~\eqref{eq:FixPointsCommuteAnsatz}, then they commute and moreover equation~\eqref{eq:FixPointCommuteEquation} is fulfilled.
\end{proof}
\noindent For example, from Observation~\ref{obs:FixPointsAnsatz} we know that if
\begin{equation}
    B=
    \begin{pmatrix}
    1 & 0 \\
    0 & 1-a
    \end{pmatrix}
    \quad\text{and}\quad
     P=\frac{1}{2}
    \begin{pmatrix}
    1+b & 0 \\
    0 & 1
    \end{pmatrix}\,,
\end{equation}
then $\bm{P}=(P,\one-P)^\dag$ is a fixed point. Note that this holds up to global changes of basis, namely for $UBU^\dag$ and $UPU^\dag$ where $U$ is unitary. For example, $P=\one/2+b\ket{+}\bra{+}$ defines a fixed point with respect to $A=\one-a\ket{-}\bra{-}$, where we denote $\ket{\pm}=1/\sqrt{2}(\ket{0}\pm\ket{1})$.

We have found fixed points for the commutative case. To tackle the noncommutative case, we will restrict to $n=d=2$ and define a linear map $\Xi$ denoting the following matrix convex combination~\cite{paulsen_MatConvex_2003,HeltonMatConv_2015},
\begin{equation}\label{eq:DefconvcombAd2}
    \Xi(B,P):=\sqrt{B}P\sqrt{B} + \sqrt{\one-B}(\one-P)\sqrt{\one-B},
\end{equation}
which is obtained in the first entry of $\Q=\B*^\dag\P$ where $\B$ and $\P$ are given in~\eqref{eq:BandP2x2} We will say that $P$ is a fixed point of $B$ if $\P$ is a fixed point of $\B$, which is equivalent to $\Xi(B,P)=P$. 
\begin{obs}\label{lem:FixPointsN=2commute}
Let $\P\in\Delta_{2,2}$ and $\B\in\mathcal{B}_{2,2}$. If $P$ is a fixed point of $B$, then $P$ commutes with $B$.
\end{obs}
\begin{proof}
Without loss of generality, let us work on the eigenbasis of $B$. Defining
\begin{equation}
    B=
    \begin{pmatrix}
    a & 0 \\
    0 & b
    \end{pmatrix}
    \quad\text{and}\quad
     P=
    \begin{pmatrix}
    p_{11} & p_{12} \\
    p_{21} & p_{22}
    \end{pmatrix}\,,
\end{equation}
algebraic manipulation shows that the equation $\Xi(B,P)=P$ can be restated as
\begin{equation}
\begin{pmatrix}
2p_{11}(a-1) & p_{12}(\sqrt{1-a}\sqrt{1-b}+\sqrt{ab}-1) \\
p_{21}(\sqrt{1-a}\sqrt{1-b}+\sqrt{ab}-1) & 2p_{22}(b-1)
\end{pmatrix}
=
\begin{pmatrix}
a-1 & 0 \\
0 & b-1
\end{pmatrix}\,.
\end{equation}
Now suppose that $P$ is nondiagonal. Then, we have
\begin{equation}
    \sqrt{1-a}\sqrt{1-b}+\sqrt{ab}-1=0\implies 1+ab-a-b=\pm(1+ab-2\sqrt{ab})\,,
\end{equation}
which has only the trivial solution $B=\one$ for the $-$ sign and $B\propto\one$ for the $+$ sign. Similar results can be shown for the product $*$, by working in the eigenbasis of $P$. Thus, we conclude that $B$ and $P$ must commute.
\end{proof}

\begin{figure}[tbp]
    \centering
    \includegraphics[scale=1]{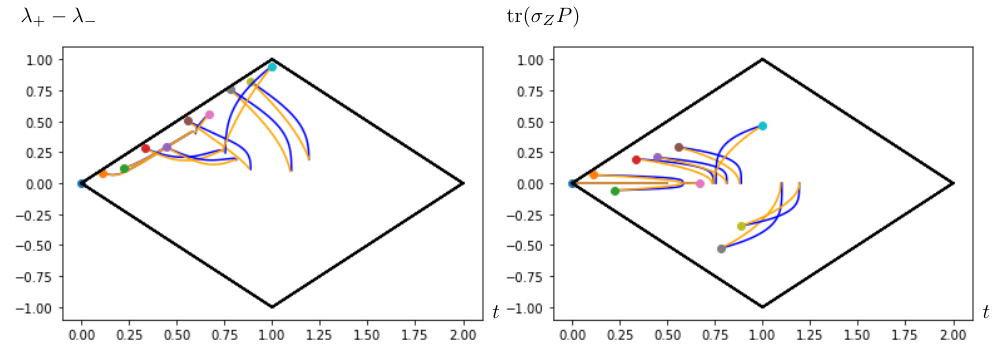}
    \caption{{\bf Blockwise bistochastic dynamics of quantum measurements in $\Delta_{2,2}$, under small noise.} Evolution of quantum measurements starting at ten random initial points $P$ (coloured dots) with different traces $0\leq\tr(P)\leq 1$, induced by a blockwise bistochastic matrix with diagonal entry $\epsilon$-close to identity, $B=\one-\epsilon\dyad{-}$, choosing $\epsilon=0.01$. Trajectories induced by the operation $\*$ ($\*^\dag$) are depicted in blue (orange). We show along the coordinate $t$,  the difference of eigenvalues $\lambda_+(P)-\lambda_-(P)$ (left) and the projection $\tr(\sigma_ZP)$ onto the Pauli operator $\sigma_Z$ (right).}
    \label{fig:DynamicsCheeseSmNoise}
\end{figure}

To finish this section we will study slow dynamics where $B$ is the identity with a real noise $\epsilon$, namely
\begin{equation}\label{eq:NoisyIdentityA}
    B=
    \begin{pmatrix}
        1-\epsilon & \epsilon \\
        \epsilon & 1-\epsilon
    \end{pmatrix}=\one-2\epsilon\ket{-}\bra{-}
\end{equation}
with $0\leq\epsilon<<1$. Now we know that in this case, any $P$ of the form $P=\one/2+\lambda\ket{+}\bra{+}$, with $-1/2\leq\lambda\leq 1/2$ to ensure positivity of both $P$ and $\one-P$, defines a fixed point. For instances of the form of Eq.~\eqref{eq:NoisyIdentityA}, numerical evidence suggests that all fixed points are of this form. The random matrices which generate the sample are defined in Eq.~\eqref{eq:GinibreRandom}. This is demonstrated in Figure~\ref{fig:DynamicsCheeseSmNoise}.
\begin{obs}
For a large sample (of the order of $10^{6}$) of Ginibre-random matrices $P$ (see~\ref{app:NonlinearMonotone}) and matrices $B$ of the form~\eqref{eq:NoisyIdentityA}, in all instances the dynamics converge to fixed points of the form of Eq.~\eqref{eq:FixPointsCommuteAnsatz} for both products $*$ and $*^\dag$.
\end{obs}

\section{Nonlinear monotone}\label{app:NonlinearMonotone}
Theorems~\ref{thm:StateIndep} and~\ref{thm:DepMajor} give a family of linear monotones for the blockwise bistochastic dynamics of sortable quantum measurements. Here we address the majorization relations beyond the linear case and ask if a nonlinear monotone can be found outside of the sortable set. 
\subsection{A conjecture for a nonlinear monotone}
To address this question, note that $\pp$ majorizes $\qq$ if and only if $(p_1-\tfrac{1}{n},...,p_n-\tfrac{1}{n})$ majorizes $q=(q_1-\tfrac{1}{n},...,q_n-\tfrac{1}{n})$, since the non-increasing order makes all cumulative sums be nonnegative. On the other hand, note that if $\pp$ majorizes $\qq$, then the cumulative sums of $\pp$ in nonincreasing order are larger or equal than the cumulative sums of $\qq$ for any possible order of $\qq$. We are now ready for the following conjecture.
\begin{conj}\label{conj:ConjectureMajorize}
Let $\bm{P}\in\Delta_{n,d}$ and $\bm{Q}\in\Delta_{n,d}$ be blockwise probability vectors. If there exists a blockwise bistochastsic matrix $\bm{B}\in\mathcal{B}_{n,d}$ such that $\bm{Q}=\bm{B}\*\bm{P}$, then for any ordering of $\{Q_i\}$ 
there exists an ordering of $\{P_i\}$ such that
\begin{equation}
    \Big |\Big |\sum_{i=1}^k\Big (P_i-\frac{1}{n}\one\Big )\Big |\Big |_2\,\geq\,
    \Big |\Big |\sum_{i=1}^k\Big (Q_i-\frac{1}{n}\one\Big )\Big |\Big |_2
\end{equation}
for all $1\leq k\leq n$, where $||A||_2=\sqrt{\tr[A^\dag A]}$ is the 2-norm.
\end{conj}
\noindent Physically, this conjectured monotone for $k=1$ can be understood in terms of the generalized Bloch vector $\tau$ introduced in Section~\ref{sec:BPVandPOVMS}, with origin at $\one/n$ and pointing to $P_1$. One has
\begin{equation}\label{eq:r(P_i)definition}
   \Big |\Big |P_1-\frac{1}{n}\one\Big |\Big |_2 = \sqrt{d\tr\Big [\Big (P_1-\frac{\one}{n}\Big )^2\Big ]} = \sqrt{\Big (\t-\frac{d}{n}\Big )^2 + \tau^2}\,,
\end{equation}
which gives the Hilbert-Schmidt distance between $P_1$ and $\one/n$, and thus it contains information about the position of $\bm{P}$ with respect to the trivial POVM described by the uniform vector $\bm{P_u}=(\frac{\one}{n},...,\frac{\one}{n})^\dag$.

Conjecture~\ref{conj:ConjectureMajorize} is supported by a sample of random blockwise probability vectors and blockwise bistochastic matrices of different sizes. This sample consists of combining different ways of choosing random blockwise probability vectors and random blockwise bistochastic matrices. Random blockwise probability vectors from $\Ddn$ have been chosen in the following way.
\begin{enumerate}
    \item Choose a square random matrix $X$ of order $d$ from the {\em Ginibre ensemble}~\cite{ginibre1965,Zyczkowski_Ginibre2011}, with entries of the form $a+ib$, where $i$ is the complex unit and $a$ ($b$) is chosen randomly from a Gaussian distribution of real numbers with mean 0 and variance 1. Assign to the first component the operator 
    \begin{equation}\label{eq:GinibreRandom}
        P_1=t_1\frac{X^\dag X}{\tr[X^\dag X]},
    \end{equation}
    where $t_1$ is a random number drawn uniformly from $[0,1]$. In this way a random positive operator $P_1$ is generated according to the Hilbert-Schmidt measure in the set of quantum states $\Omega_d$ rescaled by the factor $t_1$~\cite{Zyczkowski_Ginibre2011}.
    
    The rest of the components of $\bm{P}$ are computed using a compatibility optimizer in a semidefinite program, with the constraints that all effects $P_i$ are positive semidefinite and sum to $\one$.
    \item Take $\bm{P}$ chosen as previously and multiply every component by $\epsilon << 1$. Add to the first component $(1-\epsilon)\one$. This creates a blockwise probability vector which is $\epsilon-$close to $(\one,...,0)^T$.
    \item Take $\bm{P}$ chosen as in the first case and multiply every component by $\epsilon << 1$. Add to each component $\frac{1-\epsilon}{n}\one$. This creates a blockwise probability vector which is $\epsilon-$close to $(\one,...,\one)^T/n$.
\end{enumerate}
In a similar way, we sample blockwise bistochastic matrices constructed as follows.
\begin{enumerate}
    \item Choose $B_{11}$ according to Eq.~\eqref{eq:GinibreRandom} with $X$ a random Ginibre matrix
      and $t_1$ a random number drawn uniformly from $[0,1]$. The rest of the components $B_{ij}$ of $\bm{B}$ are computed using a compatibility optimizer in a semidefinite program, with the constraints $B_{ij\geq 0}$ and $\sum_iB_{ij}=\sum_jB_{ij}=\one$ according to Definition~\ref{def:BlockBist}.
    \item Take $\bm{B}$ chosen as previously and multiply every component by $\epsilon << 1$. Add to each diagonal component $(1-\epsilon)\one$. This creates a blockwise probability vector which is $\epsilon-$close to $\one_{nd}$, which leaves invariant any blockwise probability vector.
    \item Take $\bm{B}$ chosen as in the first case and multiply every component by $\epsilon << 1$. Add to each component $\frac{1-\epsilon}{n}\one$. This creates a blockwise probability vector which is $\epsilon-$close to $B_{ij}=\one/n$, which brings any blockwise probability vector to $(\one,...,\one)/n$.
    \item Choose the first row of $\bm{B}$ to be a blockwise probability vector constructed as in step $(i)$ above. Choose the rest of the entries of $\bm{B}$ as $B_{i,j}=B_{1,j+i}$ where the sum is taken modulo $n$. Let us mention that this is a generalization of so-called \emph{bistochastic circulant matrices}, which are defined in this way for $d=1$.
\end{enumerate}

\subsection{Analytical proof for the two-effect case in arbitrary dimensions}

We shall see that Conjecture~\ref{conj:ConjectureMajorize} holds for the case $n=2$. Before that, let us comment on three properties of this special case. First, any blockwise probability vector $\bm{P}=(P,\one-P)$ is fully determined by its first $P$. Second, one has
\begin{equation}\label{eq:r(P)=r(1-P)}
    ||\Xi(B,P)||_2=||\Xi(B,\one-P)||_2\,,
\end{equation}
from which we recover the fact that in $\Delta_2$ both components of any vector are equally distant from $1/2$. Third, let us recall that for $n=2$ both products from the right and from the left preserve the set of measurements $\Delta_{2,d}$. Accordingly, in the following Lemma we shall see that Conjecture~\ref{conj:ConjectureMajorize} holds for both products in the case $n=2$.
\begin{lemma}\label{lem:4vectorDecreases}
After the blockwise products $*$ and $*^\dag$ by a blockwise bistochastic matrix $\B\in\mathcal{B}_{2,d}$ on a blockwise probability vector $\P=(P,\one-P)^\dag$, the vector position of $P$ in the dual cone $||P||_2$ shrinks,
\begin{equation}\label{eq:ineqLem}
    ||P||_2\geq ||\Xi(B,P)||_2.
\end{equation}
\end{lemma}
\begin{proof}
We will first prove the Lemma for the product $*^\dag$, using the map $\Xi(B,P)$. By the cyclic property of the trace, the result extends to the product $*$. Starting from Eq.~\eqref{eq:ineqLem}, by using Eq.~\eqref{eq:r(P_i)definition} and squaring both sides of the equality, we arrive at the equivalent condition
\begin{equation}\label{eq:Monontone}
    \tr(\Xi(B,P)^2)-\tr(\Xi(B,P))\leq\tr(P^2)-\tr(P).
\end{equation}
By developing this expression and canceling terms out, we arrive at
\begin{equation}
    2\tr[PBPB]-4\tr[PB^2]-2\tr[BP^2]-4\tr[PB]+\tr[B^2]-\tr[B]-2\tr[(P\sqrt{B-B^2})^2]\leq 0\,.
\end{equation}
By positivity of $\one-B$, the eigenvalues of $B$ must be smaller or equal than 1, and therefore we have \begin{equation}
    \tr[B^2]-\tr[B]\leq 0\,.
\end{equation}
It will be convenient to invert the sign of the last term and express what is left to prove as
\begin{equation}
    0\geq 2\tr[PBPB]-4\tr[PB^2]-2\tr[BP^2]-4\tr[PB]+2\tr[(P\sqrt{B^2-B})^2]:=F\,,
\end{equation}
where we have defined the expression above as $F$ to shorten the notation below. Using that
\begin{equation}
    \tr[(XY)^2]\leq\tr[X^2Y^2]
\end{equation}
for $X,Y$ Hermitian matrices~\cite{Bellman1980}, we have that
\begin{equation}
\begin{aligned}
    F &\leq 2\tr[PBPB]-4\tr[PB^2]-4\tr[BP^2]-4\tr[PB]+2\tr[P^2B^2] \\
    &\leq 4(\tr[P^2B^2]-\tr[PB]-\tr[PB^2]+\tr[P^2B])=\tr[BP(PB-\one-B+P)]:=F'\,.
\end{aligned}
\end{equation}
Using that 
\begin{equation}
    \tr[XY]\leq\tr[X]\tr[Y]
\end{equation}
for $X,Y$ positive semidefinite matrices and reordering terms, we see that
\begin{equation}
    F'\leq\tr[BP](\tr[PB-B]+\tr[P-\one])\leq-\tr[BP]\tr[B+\one]\tr[\one-P]\leq 0\,
\end{equation}
as $P$, $B$ and $\one-P$ are positive semidefinite by assumption.
\end{proof}
Note that for $\Delta_{2,1}$, namely if $P$ and $B$ are real numbers $p$ and $b$ contained between 0 and 1, then the Lemma above holds trivially, as the product of any two real numbers between 0 and 1 is always smaller or equal than each of them. This implies that $|1/2-p|$ is a monotone under bistochastic dynamics, which is necessary and sufficient for the standard majorization condition $p\succ Bp$.

\section*{\large References}
\addcontentsline{toc}{chapter}{Bibliography}
\bibliographystyle{iopart-num}
\bibliography{Bibliography}{}

\end{document}